\documentclass[11pt,letterpaper]{article}

\usepackage{amsmath,amsthm,amsfonts,amssymb}
\usepackage{thm-restate}
\usepackage{fullpage}
\usepackage[utf8]{inputenc}
\usepackage[dvipsnames]{xcolor}
\usepackage{xspace,enumitem}
\usepackage[hypertexnames=false,colorlinks=true,urlcolor=Blue,citecolor=Blue,linkcolor=BrickRed]{hyperref}
\usepackage[capitalise]{cleveref}
\usepackage[OT4]{fontenc}
\usepackage{ifpdf}
\usepackage{todonotes}
\usepackage{thmtools}
\usepackage{algorithm}
\usepackage[noend]{algpseudocode}
\usepackage{algorithmicx}
\usepackage{makecell}
\usepackage{multirow}
\usetikzlibrary{patterns}

\setlist[itemize]{itemsep=1pt, topsep=2pt}
\setlist[enumerate]{itemsep=1pt, topsep=2pt}

\usepackage[margin=1in]{geometry}
\title{On Incremental Approximate Shortest Paths\\in Directed Graphs}

\date{}
\author{Adam Górkiewicz\thanks{University of Wrocław, Poland. Work partially done when the author was a student scholarship recipient at IDEAS NCBR supported by the National Science Centre (NCN) grant no. 2022/47/D/ST6/02184.} \and
Adam Karczmarz\thanks{University of Warsaw and IDEAS NCBR, Poland. Supported by the National Science Centre (NCN) grant no. 2022/47/D/ST6/02184.}}

\theoremstyle{plain}
\newtheorem{theorem}{Theorem}[section]
\newtheorem{lemma}[theorem]{Lemma}

\newtheorem{observation}[theorem]{Observation}
\newtheorem{definition}[theorem]{Definition}

\newtheorem{remark}[theorem]{Remark}

\newtheorem{invariant}[theorem]{Invariant}

\def\poly{\operatorname{poly}}
\def\polylog{\operatorname{polylog}}

\renewcommand{\O}{O}
\newcommand{\Ot}{\ensuremath{\widetilde{\O}}}

\newcommand{\eps}{\ensuremath{\epsilon}}
\newcommand{\dist}{\mathrm{dist}}

\newcommand{\wei}{w}

\newcommand{\set}[1]{\left\lbrace #1 \right\rbrace}

\newcommand{\eq}[1]{\begin{align*} #1 \end{align*}}

\DeclareMathOperator*{\rank}{rank}
\DeclareMathOperator*{\PD}{\textsc{Propagate}}

\newcommand{\rev}[1]{#1^{\textrm{R}}}

\begin{document}

\maketitle

\begin{abstract}
In this paper, we show new data structures maintaining approximate shortest paths in sparse directed graphs with polynomially bounded non-negative edge weights under edge insertions.

  We give more efficient incremental $(1+\eps)$-approximate APSP data structures that work against an adaptive adversary: a deterministic one with $\Ot(m^{3/2}n^{3/4})$\footnote{Throughout, for convenience we use the standard notation $\Ot(Y)$ as a shorthand for the expression $O(Y\polylog{n})$.} total update time and a randomized one with $\Ot(m^{4/3}n^{5/6})$ total update time.
For sparse graphs, these both improve polynomially upon the best-known bound against an adaptive adversary~\cite{KarczmarzL19}.
To achieve that, building on the ideas of~\cite{ChechikZ21,KyngMG22}, we show a near-optimal $(1+\eps)$-approximate incremental SSSP data structure for a special case when all edge updates are adjacent to the source, that might be of independent interest.

We also describe a very simple and near-optimal \emph{offline} incremental $(1+\eps)$-approximate SSSP data structure.
While online near-linear partially dynamic SSSP data structures have been elusive so far (except for dense instances), our result excludes using certain types of impossibility arguments to rule them out.
Additionally, our offline solution leads to near-optimal and deterministic all-pairs bounded-leg shortest paths data structure for sparse graphs.
\end{abstract}

\section{Introduction}
Computing shortest paths is one of the most classical algorithmic problems on graphs, with numerous applications in both theoretical and practical scenarios.
Dynamic variants of the shortest paths problem are undoubtedly among the most extensively studied dynamic graph problems.

In \emph{dynamic shortest paths} problems, the goal is to design a data structure that maintains (or provides efficient query access to) the desired information about the shortest paths in a graph $G=(V,E)$ subject to edge set updates.
If both edge insertions and deletions on $G$ are supported, the data structure is called \emph{fully dynamic}.
If only edge insertions or only edge deletions are allowed, it is called \emph{incremental} or \emph{decremental}, respectively.
Incremental and decremental settings are together referred to as \emph{partially dynamic}.
In the fully dynamic setting, one seeks data structures with low (amortized) update time, whereas in partially dynamic settings, the typical goal is to optimize total update time, i.e., the time needed to process an entire sequence of updates.

While we would ideally like to have efficient fully dynamic data structures, their existence\footnote{Or their feasibility without resorting to complex and impractical tools such as fast matrix multiplication~\cite{AbboudW14}.} is often ruled out by conditional lower bounds (e.g.,~\cite{AbboudW14, AnconaHRWW19, JinX22}).
Partially dynamic scenarios are often more approachable algorithmically, while still being very useful in applications (such as max-flow, e.g.,~\cite{BernsteinGS21, Madry10}), including designing the said fully dynamic data structures.

When analyzing randomized dynamic algorithms, one needs to specify the assumptions about how an adversary may adapt to the data structure's outputs for the data structure to work correctly and efficiently.
An \emph{adaptive} adversary can adjust the sequence of updates and queries it issues freely based on the preceding data structure's outputs.
An \emph{oblivious} adversary has to fix the entire update sequence before the process starts (but without revealing it upfront).
Adaptive data structures are more desirable; e.g., they can be applied in a black-box manner when designing static algorithms.

\subsection{Incremental shortest paths}
In this paper, we focus on \emph{incremental} shortest paths data structures for weighted \emph{directed} graphs.

\paragraph{Exact data structures.}
Even in incremental weighted digraphs, maintaining shortest paths \emph{exactly} is computationally challenging.
As proved recently by~\cite{SahaWX025}, one cannot even exactly maintain the distance between a \emph{single} fixed source-target pair in an incremental graph within $O(m^{2-\eps})$ total update time (for any density $m$), i.e., significantly faster than recomputing from scratch.
The lower bound is conditional on the Min-Weight 4-Clique Hypothesis (eg.~\cite{AbboudWW14, BackursT17, LincolnWW18}) and holds even in the \emph{offline} setting, that is, when the update sequence is revealed upfront.
A folklore observation is that the all-pairs shortest paths (APSP) matrix can be updated after an edge insertion in $O(n^2)$ time, and hence this yields an incremental APSP data structure with $O(mn^2)$ total update time.
This is optimal if the $n^2$ distances are to be maintained explicitly.

Non-trivial exact data structures are known for the special case of unweighted digraphs. Single-source shortest paths (SSSP) in unweighted digraphs can be maintained in $O(nm)$ total time~\cite{EvenS81,HenzingerK95}, whereas APSP can be maintained in $\Ot(n^3)$ total time~\cite{AusielloIMN92}.
Both these bounds are near-optimal if the distances are to be maintained explicitly.
There is also an $n^{2-o(1)}$ lower bound for maintaining a single source-target distance in an incremental sparse digraph~\cite{GutenbergWW20}. 

\paragraph{Approximate SSSP.} The lack of non-trivial approaches for exactly maintaining shortest paths in incremental weighted digraphs motivates studying approximate data structures.
For SSSP, the classical ES-tree~\cite{EvenS81,HenzingerK95} can be generalized to partially dynamic weighted digraphs with real weights in $[1,W]$ at the cost of $(1+\eps)$-approximation (where $\eps=O(1)$) and has $\Ot(mn\log(W))$ total update time (see, e.g.,~\cite{Bernstein09, Bernstein16}).
While quadratic in the general case, the approximate ES-tree is rather flexible and might be set up to run much faster -- in $\Ot(mh\log(W))$ total time -- if guaranteed that shortest (or short enough) paths from $s$ use at most $h$ hops.

Henzinger, Krinninger and Nanongkai~\cite{HenzingerKN14,HenzingerKN15} were the first to break through the quadratic bound of the ES-tree polynomially, also in the decremental setting, albeit only against an oblivious adversary.
Probst Gutenberg, Vassilevska Williams, and Wein~\cite{GutenbergWW20} showed the first SSSP data structure designed specifically for the incremental setting against an adaptive adversary, with $\Ot(n^2\log(W))$ total update time.
Their data structure is deterministic and near-optimal for dense graphs.
Chechik and Zhang~\cite{ChechikZ21} showed two data structures for sparse graphs: one deterministic with $\Ot(m^{5/3}\log(W))$ total update time and another Monte Carlo randomized with $\Ot((mn^{1/2}+m^{7/5})\log(W))$ total update time, correct w.h.p.\footnote{With high probability, that is, with probability $1-n^{-c}$ for any desired constant $c\geq 1$.} against an adaptive adversary.
These two bounds were subsequently improved to $\Ot(m^{3/2}\log(W))$ (deterministic) and $\Ot(m^{4/3}\log(W))$ (randomized) respectively by Kyng, Meierhans and Probst Gutenberg~\cite{KyngMG22}.

Very recently,~\cite{CKLMG24} gave a deterministic $(1+\eps)$-approximate incremental min-cost flow algorithm implying that a $(1+\eps)$-approximate shortest path for a \emph{single} source-target pair in polynomially weighted digraphs can be maintained in almost optimal $m^{1+o(1)}$ time.

No non-trivial conditional lower bounds for incremental approximate SSSP have been described and thus the existence of an SSSP data structure with near-linear total update time remains open.

\paragraph{Approximate APSP.} Bernstein~\cite{Bernstein16} showed a Monte Carlo randomized data structure maintaining $(1+\eps)$-approximate APSP in either incremental or decremental setting with $\Ot(mn\log(W))$ total update time.
This bound is near-optimal (even for incremental transitive closure~\cite{HenzingerKN15}), but the data structure gives correct answers w.h.p. only against an oblivious adversary.
A deterministic $(1+\eps)$-approximate data structure with $\Ot(n^3\log(W))$ total update time is known~\cite{KarczmarzL20}.
This is near-optimal for dense graphs~\cite{HenzingerKN15}. \cite{KarczmarzL19} showed a deterministic data structure with $\Ot(mn^{4/3}\log^2(W))$ total update time that is more efficient for sparse graphs.

Interestingly, the incremental APSP data structures for sparse digraphs~\cite{Bernstein16,KarczmarzL19} are both obtained without resorting to any of the critical ideas from the state-of-the-art sparse incremental SSSP data structures~\cite{KyngMG22}.
Instead, they work by running approximate ES-trees on carefully selected instances where approximately shortest paths from the source have sublinear numbers of hops.
Note that a deterministic bound $\Ot(nm^{3/2}\log(W))$ and a randomized adaptive bound $\Ot(nm^{4/3}\log(W))$ could be achieved by simply employing the two data structures of \cite{KyngMG22} for all possible sources $s\in V$.
This black-box application does not improve upon~\cite{KarczmarzL19},~though.
It is thus interesting to ask whether the data structures of~\cite{ChechikZ21,KyngMG22}, just like the ES-tree, could be adjusted to run faster in some special cases that are useful in designing all-pairs data structures.

\subsection{Our contribution}
In this paper, we give new incremental $(1+\eps)$-approximate shortest paths data structures for sparse directed graphs.
With the goal of obtaining improved all-pairs data structures, we first identify a useful special case of the incremental SSSP problem for which a data structure with near-linear total update time is achievable.
Building on the ideas from~\cite{ChechikZ21,KyngMG22}, we prove:
\begin{theorem}[Simplified version of Theorem~\ref{t:source-sssp}]\label{t:source-sssp-intro}
Let $\eps\in (0,1)$. Let $G=(V,E)$ be a digraph with edge weights in ${\{0\}\cup [1,W]}$ and a source~$s\in V$. There exists a deterministic data structure explicitly maintaining distance estimates $d:V\to\mathbb{R}_{\geq 0}$ satisfying
\[ \dist_G(s,v)\leq d(v)\leq (1+\eps)\cdot\dist_G(s,v) \]
for all $v\in V$ and only supporting insertions (or weight decreases) of \emph{source edges} $e=sv$, $v\in V$.

The total update time of the data structure
  is $O(m\log{(nW)}\log^2(n)/\eps+\Delta)$, where $m$ is the final number of edges in~$G$
and $\Delta$ is the total number of updates issued.
\end{theorem}
Even though the repertoire of possible updates in~\Cref{t:source-sssp-intro} seems very limited, we remark that in the exact setting, explicitly maintaining distances under source updates like this requires $\Omega(mn)$ total time since one could reduce static APSP to a sequence of $O(n)$ updates of this kind.\footnote{This is precisely the argument used in~\cite{RodittyZ11} to reduce static APSP to partially dynamic SSSP.}

We also show (\Cref{l:arbitrary-extension}) that without compromising the total update time,~\Cref{t:source-sssp-intro} can be extended -- at the cost of slight accuracy decline -- to also accept arbitrary edge insertions, as long as they do not improve the distances from the source too much.
With this observation, we can combine~\Cref{t:source-sssp-intro} with the near-optimal APSP data structure for dense digraphs~\cite{KarczmarzL20} and obtain improved incremental APSP data structure for sparse graphs against an adaptive adversary:
\begin{restatable}{theorem}{tallpairs}\label{t:all-pairs}
Let $G=(V,E)$ be a digraph with edge weights in $\{0\}\cup [1,W]$ and let $\eps\in (0,1)$. There exist data structures explicitly
maintaining distance estimates $d:V\times V\to\mathbb{R}_{\geq 0}$ such that
\[ \dist_G(u,v)\leq d(u,v)\leq (1+\eps)\cdot\dist_G(u,v) \]
for all $u,v\in V$,
subject to edge insertions or weight decreases issued to $G$:
\begin{itemize}
  \item A deterministic one with $O(m^{3/2}n^{3/4}\log^2(nW)\log(n)/\eps^{3/2}+\Delta)$ total update time.
  \item A Monte Carlo randomized one that produces correct answers against an adaptive adversary (whp.)
    with $O(m^{4/3}n^{5/6}\log^3(nW)\log(n)/\eps^{7/3}+\Delta)$ total update time.
\end{itemize}
Here, $m$ denotes the final number of edges in $G$ and $\Delta$ the total number of updates issued.
\end{restatable}
For sparse graphs, the deterministic variant improves upon the total update time of the data structure of~\cite{KarczmarzL19} by a factor of $n^{1/12}$.
By employing the randomized technique of preventing error buildup due to~\cite{ChechikZ21}, the speed-up over~\cite{KarczmarzL19} is increased to a factor of $n^{1/6}$.

We believe that~\Cref{t:source-sssp-intro} and its extensions may find other applications in the future. 

\paragraph{Offline data structures.} Finally, we also prove that the $(1+\eps)$-approximate incremental SSSP problem has a near-linear deterministic \emph{offline} solution. Formally, we show:

\begin{restatable}[]{theorem}{restateOfflineTheorem}\label{t:offline}
Let $\eps\in (0,1)$. Let $G=(V,E)$ be a digraph with edge weights in $\{0\}\cup [1,W]$ and a source $s\in V$.
Suppose we are given a sequence of $\Delta$ incremental updates to $G$ (edge insertions or weight decreases). Let $G^0,\ldots,G^\Delta$ denote the subsequent versions of the graph $G$.

Then in
$\O(m \log(n)\log(nW) \log^2(\Delta) / \eps + \Delta)$
deterministic time one can compute a data structure~supporting queries $(v,j)$ about a $(1+\eps)$-approximate estimate of $\dist_{G^j}(s,v)$. The query time is $\O(\log(\log(nW)\log(\Delta)/\eps))$.
For example, if $\eps,W,\Delta\in \poly(n)$, the query time is $O(\log\log{n})$.
\end{restatable}
Note that incremental and decremental settings are equivalent in the offline scenario, so the above works for decremental update sequences as well.
To the best of our knowledge, no prior\footnote{An independent and parallel work~\cite{arxiv-independent} established a similar result on offline incremental SSSP.}  work described offline partially dynamic $(1+\eps)$-approximate SSSP algorithms. It is worth noting that for dense graphs, the best known online data structures~\cite{BernsteinGW20, GutenbergWW20} already run in near-optimal $\Ot(n^2\log(W))$ time.
However, the state-of-the-art partially dynamic SSSP bounds~\cite{BernsteinGW20, KyngMG22, GutenbergWW20} are still off from near-linear by polynomial factors.

\Cref{t:offline} implies that if incremental or decremental approximate SSSP happen to not have near-linear algorithms, then showing a (conditional) lower bound for the problem requires exploiting the online nature of the problem.
For example, showing reductions from such problems as the Min-Weight $k$-Clique or $k$-cycle detection (as used in~\cite{SahaWX025} and~\cite{GutenbergWW20}, respectively, for exact shortest path) cannot yield non-trivial lower bounds for partially dynamic SSSP.

Offline incremental APSP and its special cases have been studied before under the name \emph{all-pairs shortest paths for all flows} (APSP-AF)~\cite{DuanR18, ShinnT14}, or -- in a special case when edges are inserted sorted by their weights -- \emph{all-pairs bounded-leg shortest paths} (APBLSP)~\cite{BoseMNSZ04, DuanP08, RodittyS11}.
For polynomially weighted digraphs, all known non-trivial data structures are $(1+\eps)$-approximate.
Duan and Ren~\cite{DuanR18} described a deterministic data structure with $\Ot(n^{(\omega+3)/2}\log(W))$ preprocessing and $O(\log\log(nW))$ query time.\footnote{Here, $\omega$ denotes the matrix multiplication exponent. Currently, it is known that $\omega\leq 2.372$ \cite{AlmanDWXXZ25}.}
For sparser graphs with polynomial weights, the randomized partially dynamic data structure of~\cite{Bernstein16} (augmented with persistence~\cite{DriscollSST89} to enable answering queries about individual graph's versions) has $\Ot(mn)$ total update time and remains the state of the art, also for APBLSP.
By applying~\Cref{t:offline} for each source, one obtains an APSP-AF data structure that is much simpler, deterministic, and a log-factor faster than~\cite{Bernstein16}.

Finally, let us remark that even though, as presented in this paper, all the data structures we develop only maintain or output distance estimates, they can be easily extended to enable reporting a requested approximately shortest path $P$ in near-optimal $\Ot(|P|)$ time.

\subsection{Further related work}
There is a broad literature on dynamic shortest paths data structures for weighted digraphs in fully dynamic (e.g.,~\cite{AlokhinaB24, BrandN19, Bernstein16, ChechikZ23, DemetrescuI04, Mao24a}) and decremental (e.g., \cite{Bernstein16, BernsteinGS20, BernsteinGW20, KarczmarzL19}) settings.
Much of the previous (and in particular very recent) work has been devoted to shortest paths data structures for undirected graphs, often with larger approximation factors, e.g.,~\cite{Bernstein09, BernsteinGS21, Chechik18, DoryFNV24, ForsterGNS23, forsterNG23, HaeuplerLS24, KyngMG24}.

\subsection{Organization}
In~\Cref{sec:prelims} we set some notation. Sections~\ref{sec:propagate}~and~\ref{sec:source-insertions} together contain the description of the data structure for source-incident insertions and its extensions.
In~\Cref{sec:all-pairs} we give our incremental APSP data structure.
Finally,~\Cref{sec:offline} describes the offline incremental SSSP data structure.

\section{Preliminaries}\label{sec:prelims}
In this paper, we only deal with non-negatively weighted digraphs $G=(V,E)$.
We write $e=uv$ to denote a directed edge from its tail $u$ to its head $v$. We denote by $\wei(e)$ the weight of edge~$e$.
By $G+F$ we denote the graph $G$ with edges $F$ added.
We also write $G+e$ to denote $G+\{e\}$.

We define a path $P=s\to t$ in $G$ to be a sequence of edges $e_1\ldots e_k$, where $u_iv_i=e_i\in E$, such that $v_i=u_{i+1}$, $u_1=s$, $v_k=t$.
If $s=t$, then $k=0$ is allowed.
Paths need not be simple, i.e., they may have vertices or edges repeated.
We often view paths as subgraphs and write~$P\subseteq G$.

The weight (or length) $\wei(P)$ of a path $P$ equals the sum $\sum_{i=1}^k\wei(e_i)$.
We denote by $\dist_G(s,t)$ the weight of the shortest $s\to t$ path in $G$. 
If the graph in question in clear from context, we sometimes omit the subscript and write $\dist(s,t)$ instead.

\section{Dijkstra-like propagation}\label{sec:propagate}
\newcommand{\epsp}{\xi}
\newcommand{\Vinput}{V_{\mathrm{input}}}
\newcommand{\Vtouch}{V_{\mathrm{touched}}}
In this section we define and analyze properties of a Dijkstra-like procedure $\PD$ whose variant has been
used by~\cite{ChechikZ21,KyngMG22} for obtaining state-of-the-art incremental SSSP data structures. We will use that in our specialized incremental SSSP data structures in Section~\ref{sec:source-insertions}.

Let $G=(V,E)$ be a weighted directed graph with edge weights in $\{0\}\cup[1,W]$.
Let $\epsp\in (0,1)$ be an accuracy parameter, and let us put $\ell:=\lceil{\log_{1+\epsp}(nW)}\rceil$.

Let $d:V\to \{0\}\cup [1,nW]$ be some \emph{estimate vector}. We would like the estimates to satisfy and maintain -- 
subject to edge insertions (or weight decreases) issued to $G$ --
the following invariant:
\begin{invariant}\label{slack-inv}
For every edge $uv=e\in E$, $d(v)\leq (1+\epsp)(d(u)+\wei(e))$.
\end{invariant}
The basic purpose of the procedure $\PD$ 
(see Algorithm~\ref{alg:propagate-additive}) is to decrease some of the estimates $d(\cdot)$ so that~\Cref{slack-inv}
(potentially broken by edge insertions) is again satisfied.

Roughly speaking, $\PD$ propagates estimate updates similarly to Dijkstra's algorithm, but only decreases a given estimate if the new value is smaller than the old one by a factor at least $1+\epsp$ or if the vertex it corresponds to is currently in the queue.
For $\Vinput \subseteq V$, $\PD(\Vinput)$ initializes the queue with vertices from the set $\Vinput$.
It returns a subset $\Vtouch \subseteq V$ of vertices that have been explored (i.e., inserted to the queue) at any point during the procedure's run.

\begin{observation}\label{obs:insert}
Suppose an edge  $e=uv$ is inserted to $G$. If $d(v)>(1+\epsp)(d(u)+\wei(e))$,~then setting $d(v):=d(u)+\wei(e)$ followed by
calling $\PD(\{v\})$ ensures \Cref{slack-inv} is satisfied.
\end{observation}
\begin{algorithm}[b]
\caption{$\PD(\Vinput)$}\label{alg:propagate-additive}
\begin{algorithmic}[1]
  \State Initialize a priority queue $Q=\Vinput$ keyed with $d(\cdot)$, and a set $\Vtouch:=\Vinput$.
  \While{$Q\neq \emptyset$}
    \State $u:=Q.\text{pop}()$
    \For{$uv=e\in E$}
      \If{$v\in Q$}
        \State $d(v):=\min(d(v),d(u)+\wei(e))$ \Comment{decrease key}
      \ElsIf{$d(v)>(1+\epsp)(d(u)+\wei(e))$}
        \State $d(v):=d(u)+\wei(e)$
        \State Insert $v$ to $Q$ and $\Vtouch$
      \EndIf
    \EndFor
  \EndWhile
  \Return $\Vtouch$
\end{algorithmic}
\end{algorithm}
Our data structures later on will rely on the following stronger property of $\PD$.

\begin{restatable}{lemma}{lemnoslack}\label{lem:no_slack}
After running $\PD(\Vinput)$, for every edge $uv=e\in E$ such that \linebreak ${u, v \in \Vtouch}$, $e$ is \emph{relaxed}, i.e., we have $d(v)\leq d(u)+\wei(e)$.
\end{restatable}
  \begin{proof}
    Note that for any $x\in \Vtouch$, after $x$ is popped from $q$, $d(x)$ will not change any further until $\PD$ completes 
    since all the estimate changes that follow will set the estimates to values no less than $d(x)$.
		Consider the moment when $u$ is popped from the queue $Q$ in $\PD$.

  \begin{enumerate}[label=${\arabic*}^\circ)$]
			\item $v$ was popped from $Q$ before $u$. In that case $d(v) \le d(u)$, thus
        $e$ is currently relaxed and this will not change in the current $\PD$ call, as desired.
      \item $v$ is currently in $Q$. Then observe that we set $d(v):=\min\set{d(v), d(u) + \wei(e)}$. Thus we have $d(v) \le d(u) + w(e)$ at that point and $d(v)$ will not increase, so $e$ will remain relaxed till the end of $\PD$.
			\item $v$ has not been pushed to $Q$ yet. In particular, $v\notin \Vinput$.

  \begin{enumerate}[label=$3.{\arabic*}^\circ)$]
				\item
          $d(v)>(1+\epsp)(d(u)+\wei(e))$. In that case when relaxing the edge $e$, we set $d(v):=d(u) + w(e)$ which
          makes $e$ relaxed.
          $d(v)$ can only decrease, and thus $e$ will remain relaxed.
				\item
          $d(v)\leq (1+\epsp)(d(u)+\wei(e))$.
          Note that when $v$ is pushed to $Q$ later on, $d(v)$ must decrease by a factor of $1 + \epsilon$, thus the estimate $d'(v)$ of $v$
          when $\PD$ completes will satisfy
          \[(1 + \epsp)d'(v) < d(v).\]
					This, combined with
          our assumption, yields:
          \[d'(v)<d(u)+\wei(e). \]
          Since $d(u)$ will not decrease any further,
          $e$ will be relaxed at the end of $\PD$.
			\end{enumerate}
		\end{enumerate}
	\end{proof}

We call a path $P=u\to v$ in $G$ \emph{relaxed} if $d(v)\leq d(u)+\wei(P)$. Note that if all edges of $P$ are relaxed then by chaining inequalities for the individual edges we obtain that $P$ is relaxed as well.

The following
bounds the total time needed to maintain the estimates subject to edge insertions.
\begin{lemma}\label{l:charge-decrease}
If edge insertions issued to $G$ are processed using~\Cref{obs:insert}, then the total time spent on maintaining the estimates $d(\cdot)$ is $O(n\ell\log{n}+m\ell)=O((m+n\log{n})\log{(nW)}/\epsp)$.
\end{lemma}
\begin{proof}
A vertex $v$ is pushed to the queue $Q$ of $\PD$ only after $d(v)$ drops by a factor at least $1+\epsp$. This may happen at most $\ell=O(\log(nW)/\epsp)$ times per vertex since $d(v)\leq nW$ initially and either $d(v)\geq 1$ or $d(v)=0$. Processing $v$ every time it is popped (in $O(\log{n})$ time) from the queue $Q$ costs $O(\deg(v))$ time.
\end{proof}

\section{An SSSP data structure for source insertions}\label{sec:source-insertions}
In this section, inspired by the deterministic algorithm of~\cite{KyngMG22}, we describe an incremental SSSP data structure supporting only a very limited type of edge insertions, namely, insertions of edges originating directly in the source $s$.
We first show a basic variant and then also describe some extensions that will be required in our incremental APSP data structure.
We show:
\begin{theorem}\label{t:source-sssp}
Let $\epsp\in (0,1)$. Let $G=(V,E)$ be a digraph with weights in ${\{0\}\cup [1,W]}$ and~a~source $s\in V$. There exists a data structure explicitly maintaining estimates $d:V\to\mathbb{R}_{\geq 0}$ satisfying:
\begin{itemize}
  \item $d(v)$ is the weight of some $s\to v$ path in $G$, ie., $\dist_G(s,v)\leq d(v)$, and
  \item $d(v) \leq (1+\epsp)^{\log_2{m}+1}\dist_G(s,v)$
\end{itemize}
for all $v\in V$.
The data structure supports insertions (or weight decreases) of edges $e=sv$, $v\in V$.

The total update time of the data structure
is $O(m\log{(nW)}\log{n}/\epsp+\Delta)$, where $m$ is the final number of edges in~$G$
and $\Delta$ is the total number of updates issued.
\end{theorem}

\begin{remark}\label{remark: epsp fix}
In~\Cref{t:source-sssp}, we bound the multiplicative error by $(1+\epsp)^{\log_2{m}+1}$ merely for
convenience: $(1+\epsp)$ corresponds to the ``multiplicative slack'' on each edge maintained
in~\Cref{slack-inv}.
  To achieve a target approximation factor of $1+\eps$ as in~\Cref{t:source-sssp-intro}, it is enough to set $\epsp:=\eps/(2\log_2{m}+2)$.\footnote{
  This follows by the inequality $(1+x/(2k))^k\leq e^{x/2}\leq 1+x$ that holds for all $x\in (0,1)$.}
The total update time in terms of $\eps$ then becomes $O(m\log(nW)\log^2{n}/\eps+\Delta)$.
\end{remark}

\subsection{Basic setup}\label{sec:setup}
For simplicity, let us assume that $G$ initially contains $n$ edges $sv$ of weight $nW$, so that every
vertex is reachable from $s$ and $\dist(s,v)\leq nW$. Note that unless $v$ is not reachable from $G$,
this assumption does not ever influence $\dist(s,v)$. To drop the assumption, we might run a separate
incremental single-source reachability data structure with linear total update time (such as~\cite{Italiano86})
and potentially return $\infty$ instead of the maintained $d(v)$ if $v$ is not (yet) reachable from $s$.

The estimates $d(v)$ the data structure maintains come from $\{0\}\cup [1,nW]$. 
They are initialized to $d(v):=\dist(s,v)$ using Dijkstra's algorithm, so that $d(s)=0$ and $d(v)\leq nW$.
The data structure's operations will guarantee that the estimates never decrease and, moreover,
each $d(v)$ always corresponds to the length of some $s\to v$ path in $G$.
Consequently, we will have $d(s)=0$ and $\dist(s,v)\leq d(v)$ for $v\in V$ at all times.

Let again ${\ell:=\lceil \log_{1+\epsp}(nW)\rceil}$. 
Crucially, the maintained estimates $d(\cdot)$ also obey~\Cref{slack-inv}.
\subsection{Handling source insertions}

Consider the insertion pseudocode shown in \Cref{alg:insert}.
When an edge $e=sv$ is inserted, the insertion is initially handled as explained
in~\Cref{obs:insert}.
This alone guarantees that~\Cref{slack-inv} is satisfied and costs $O(n\ell\log{n}+m\ell)$ time through all insertions (see~\Cref{l:charge-decrease}).
Afterwards, further steps (lines~\ref{line:recertify}~and~\ref{line:synchronize} in~\Cref{alg:insert}) are performed to reduce estimate errors accumulated on paths in $G$.
Roughly speaking, these steps will interact with the estimates only via some additional $\PD$ calls. 
As a result, they cannot break~\Cref{slack-inv} and thus we will assume it remains satisfied throughout.
In the following, we explain the ideas behind these steps in detail.

\begin{algorithm}[t]
\caption{\textproc{Source-Insert}($e=sv$)}\label{alg:insert}
\begin{algorithmic}[1]
  \State{$E := E \cup \set{e}$}
  \If{$d(v) > (1 + \epsp)\wei(e)$}
    \State{$d(v) := \wei(e)$}
    \State{$\Vtouch := \PD(\set{v})$} \label{line:propagation}
    \State Increase the ranks of vertices $\Vtouch$ to $r+1$\label{line:recertify}
    \State{$\textsc{Synchronize}()$}\label{line:synchronize}
  \EndIf
\end{algorithmic}
\end{algorithm} 
\begin{definition}[certificates]\label{def:cert}
Let $k\geq 0$ be an integer.
We say that a vertex $u \in V$ is $k$-certified (or has a $k$-certificate) if for every path $P = u \to v$ in $G[V\setminus \{s\}]$ we have
\[d(v)\leq (1+\epsp)^k\cdot (d(u)+\wei(P)).\] 
\end{definition}
As there exist no paths from $s$ in $G[V\setminus\{s\}]$, $s$ is trivially $k$-certified for any $k$. We now state
some properties vertex certificates' behavior after running $\PD$.
\begin{restatable}{observation}{ocertificatesremainvalid}\label{certificates_remain_valid}
Let $\Vinput\subseteq V$ be an arbitrary subset. Suppose $\Vtouch:=\PD(\Vinput)$ is run.
Then, for every $u\in V \setminus \Vtouch$, if $u$ was $k$-certified before, then it remains $k$-certified.
\end{restatable}
\begin{proof}
  Let $v\in V$. Let $d'(u),d'(v)$ be the respective estimates of $u$ and $v$ before running $\PD$. For any path $P=u\to v$ in $G[V\setminus\{s\}]$, we have $d'(v)\leq (1+\epsp)^k(d'(u)+\wei(P))$.
  But since $u\notin \Vtouch$, we have $d(u)=d'(u)$ and $d(v)\leq d'(v)$. Hence,
  $d(v)\leq (1+\epsp)^k(d(u)+\wei(P))$ also holds after $\PD$ completes.
\end{proof}
\begin{restatable}{lemma}{lrecertification}\label{recertification_theorem}
Suppose~\Cref{slack-inv} is satisfied.
Let $U \subseteq V$ be some set of $k$-certified vertices.
After running ${\Vtouch:=\PD(V \setminus U)}$, the vertices $\Vtouch$ are $(k + 1)$-certified.
\end{restatable}
\begin{proof}
Let $u\in \Vtouch$ and consider any path $P = u \to t$ in $G[V\setminus\{s\}]$, where $t\in V$.
  If $V(P)\subseteq \Vtouch$, then by \Cref{lem:no_slack}, we have 
  $d(t)\leq d(u)+\wei(P)\leq (1+\epsp)^{k+1}(d(u)+\wei(P))$.
  Otherwise, let $v$ be the first vertex on $P$ from outside $\Vtouch$.
  Note that $u\neq v$ and since $V\setminus U\subseteq \Vtouch$, $v\in U$.
  Let $v'$ be the vertex preceding $v$ on $P$.
  Express $P$ as $P_1P_2P_3$, where $P_1=u\to v'$, $P_2$ is a single-edge $v'\to v$ path, and $P_3=v\to t$.
  Consider each of the paths $P_1,P_2,P_3$ individually:
  \begin{itemize}
    \item Since $V(P_1)\subseteq \Vtouch$, by~\Cref{lem:no_slack}, we have $d(v')\leq d(u)+\wei(P_1)$.
    \item By~\Cref{slack-inv} and $v'v\in E$, $d(v)\leq (1+\epsp)(d(v')+\wei(P_2))$.
    \item Since $v\in U\setminus \Vtouch$, by the assumption and~\Cref{certificates_remain_valid}, $d(t)\leq (1+\epsp)^k(d(v)+\wei(P_3))$.
  \end{itemize}
  By chaining the inequalities, we obtain:
  \[
    d(t)\leq (1+\epsp)^k\cdot ((1+\epsp)(d(u)+\wei(P_1)+\wei(P_2))+\wei(P_3))
        \leq (1+\epsp)^{k+1}\cdot (d(u)+\wei(P)),
  \]
  as desired.
\end{proof}

The pursued data structure additionally tracks information about vertex certificates.
Since that a vertex $v$ is $k$-certified implies it is $(k + 1)$-certified as well,
for each vertex $v \in V$ we will only care about the smallest value of $k$ that we can deduce.
Initially (before updates come), each vertex is $0$-certified since all edges are relaxed.
Note that due to an insertion of an edge from $s$ itself alone, no vertex stops being $k$-certified as~\Cref{def:cert} is concerned about paths in $G[V\setminus\{s\}]$.
Moreover, if an insertion of an edge is handled as in~\Cref{obs:insert}, then by~\Cref{certificates_remain_valid},
a vertex $u \in V$ can only stop being $k$-certified when $u \in \Vtouch$ after running $\PD$.
We will show that~\Cref{recertification_theorem} is useful for controlling this effect and keeping vertex certificates small.

For each $v \in V$, the data structure additionally maintains an integer $\rank(v)\geq 0$, such that~$v$ is $\rank(v)$-certified.
Initially, $\rank(v) = 0$ for all vertices $v\in V$.
Let $C_k \subseteq V$ denote the current set of vertices $v$ with $\rank(v) = k$ and let $r$ denote the \emph{current maximum rank}.
The data structure explicitly maintains these sets alongside the information about their total degree $\deg(C_k)=\sum_{v\in C_k}\deg(v)$.
This additional bookkeeping is easy to implement subject to vertices' ranks changes and will be mostly neglected in the following description and analysis.

As discussed previously, running $\PD$ in line~\ref{line:propagation} of~\Cref{alg:insert} invalidates the (at most) $r$-certificates of vertices from $\Vtouch$.
However, by~\Cref{recertification_theorem}, for each $v \in \Vtouch$ we can fix this by assigning $\rank(v) := r + 1$ which
increases the maximum rank $r$ by one.
Doing this alone would cause $r$ to grow linearly with the number of $\PD$ invocations made during the algorithm
and consequently make the distance estimates very inaccurate.
In order to negate this effect, we will use the $\textsc{Synchronize}$ procedure (called in line~\ref{line:synchronize} in~\Cref{alg:insert}) shown in \Cref{alg:synchronize}.

\begin{algorithm}
\caption{$\textsc{Synchronize}()$}\label{alg:synchronize}
\begin{algorithmic}[1]
  \While{there exists $k \in \set{0, \dots, r}$ such that $\deg(C_k) \le \deg(C_{k + 1}) + \dots + \deg(C_{r})$}
    \State{$\Vtouch:= \PD(C_{k} \cup C_{k + 1} \cup \dots \cup C_{r})$} \label{line:synchronized_propagation}
    \State Set the ranks of vertices $\Vtouch$ to $k$
  \EndWhile
\end{algorithmic}
\end{algorithm}

\begin{observation}\label{sync_correct}
$\textsc{Synchronize}$ correctly reassigns ranks.
\end{observation}
\begin{proof}
Note that the vertices $C_0\cup\ldots\cup C_{k-1}$ are $(k-1)$-certified.
As a result, by~\Cref{recertification_theorem}, after running $\PD(C_k\cup\ldots\cup C_r)$,
we have that vertices $\Vtouch$ are $k$-certified and thus can be assigned ranks $k$.
By~\Cref{certificates_remain_valid}, there is no need to change the ranks of $V\setminus\Vtouch$.
\end{proof}
\begin{lemma}\label{l_is_log}
When $\textsc{Synchronize}$ completes, the maximum rank $r$ satisfies $r\leq \log_2 m$.
\end{lemma}
\begin{proof}
When $\textsc{Synchronize}$ completes, we have
\begin{equation}\label{eq:sync-cond}
\deg(C_k) > \deg(C_{k + 1}) + \dots + \deg(C_{r})
\end{equation}
for all $k \in \set{0, \dots, r}$.
We will show by induction that
$ C_{r - i} \ge 2^i $
for all $i = 0, \dots, r$.
For $i = 0$, we have $C_{r} \ge 0$ trivially.
For $i \ge 1$, by~\eqref{eq:sync-cond} applied to $k=r-i$ and the inductive assumption we get:
\[ \deg(C_{r - i}) > \deg(C_{r - i + 1}) + \dots + \deg(C_{r}) \ge 2^{i - 1} + \dots + 1 = 2^i - 1, \]
so $\deg(C_{r - i}) \ge 2^i$, finishing the proof.
Note that for $i = r$, we get $\deg(C_0) \ge 2^{r}$ and since $\deg(C_0) \le m$, we get $2^r \le m$.
\end{proof}

\Cref{sync_correct} and \Cref{l_is_log} show that by using synchronization, the ranks of the vertices will be correct and will be kept bounded by $\log_2 m$.
\begin{lemma}\label{l:rank}
For any $t\in V$, we have
$ d(t)\leq (1+\epsp)^{r+1}\cdot \dist(s,t)=(1+\epsp)^{1+\log_2{m}}\cdot \dist(s,t).$
\end{lemma}
\begin{proof}
For $t=s$, the lemma is trivial since $d(s)=0$ holds always.
Let $t\neq s$.
Consider a simple shortest path $P$ from $s$ to $t$ and let $u \in V$ be the second vertex on $P$.
We have $P=e\cdot P'$, where $e=su$, $P'=u\to v$ and $P'\subseteq G[V\setminus \{s\}]$.
By \Cref{slack-inv}, we have ${d(u)\leq (1+\epsp)(d(s)+\wei(e))}$.
By the definition of $\rank(u)$, we have $d(t)\leq (1+\epsp)^{\rank(u)}\cdot (d(u)+\wei(P'))$.
As $\rank(u)\leq r$, we get:
\begin{align*}
  d(t)&\leq (1+\epsp)^{\rank(u)} ((1+\epsp)(d(s)+\wei(e))+\wei(P'))
  \leq (1+\epsp)^{\rank(u)+1}\wei(P)
  =(1+\epsp)^{r+1}\dist(s,t).
\end{align*}
Finally, $d(t)\leq (1+\epsp)^{1+\log_2{m}}\cdot\dist(s,t)$ follows by combining the above with~\Cref{l_is_log}.
\end{proof}

\paragraph{Running time analysis.} 
First of all, observe that all edge updates that do not immediately
decrease some estimate by a factor of at least $1+\epsp$ are processed in constant time.
The total cost of processing such updates is hence $\O(\Delta)$.

It is easy to see that, apart from the above, the running time is dominated by the total time of the $\PD$ calls. 
In fact, the total update time can be (coarsely) bounded by summing the degrees of the vertices that are pushed to the queue $Q$ in $\PD$.
Similarly as in~\Cref{sec:propagate}, we can bound the total cost incurred by pushing a vertex after its estimate drops by
\[ \O(m\ell+n\ell\log{n})=\O((m+n\log{n})\log(nW)/\epsp).\]

Note that the only other way a vertex $v$ can be pushed to the queue in $Q$ (without a drop in the estimate $d(v)$) is if $v$ is in the input set of the $\PD$ call inside $\textsc{Synchronize}$.
Thus, consider some call $\PD(C_k \cup C_{k + 1} \dots \cup C_{r})$ inside $\textsc{Synchronize}$ for $k \in \set{0, \dots, r}$ such that
  $\deg(C_k) \le \deg(C_{k + 1}) + \dots + \deg(C_r)$.
Observe that in such a case we can bound the cost of processing the vertices initially pushed to the queue $Q$ by $\O(\log{n})$ times
  \[ \deg(C_k) + \deg(C_{k + 1}) + \dots + \deg(C_r) \le 2 \big( \deg(C_{k + 1}) + \dots + \deg(C_r) \big).\]
Recall that after such a $\PD$ call completes, the ranks of all vertices
$C_{k+1}\cup\ldots\cup C_r$ drop by at least $1$. We can thus bound the
total sum of expressions $\deg(C_{k+1})+\ldots+\deg(C_r)$ throughout by $m$ times the maximum number of times some vertex $v$ may have its rank \emph{increased}.
However, note that the rank of a vertex $v$ can only increase if $v\in\Vtouch$ in line~\ref{line:recertify} of~\Cref{alg:insert}.
But the presence of $v$ in $\Vtouch$ in that case is caused by $d(v)$ dropping by a factor of at least $1+\epsp$.
We conclude that $\rank(v)$ may only increase $\O(\ell)$ times.
As a result, the total cost $\PD$ calls in $\textsc{Synchronize}$ can
also be bounded by $\O(m\ell\log{n})=\O(m\log(nW)\log(n)/\epsp)$.

The above analysis combined with~\Cref{l:rank} yields~\Cref{t:source-sssp}.
\subsection{Handling general edge insertions in a special case}\label{sec:reset}
Recall that the data structure of~\Cref{t:source-sssp} does not allow edge insertions except
for edges that originate in the source. 
In our later developments~(\Cref{sec:all-pairs}), we need to insert batches $F$ of arbitrary edges
to the data structure from time to time.
One way to handle this would be to reinitialize the data structure for the graph $G+F$ after each such insertions batch.
This would be too costly to yield non-trivial applications though.

We nevertheless prove that if an additional assumption holds that
the maintained estimates for~$G$ are reasonably accurate for $G+F$ before applying the insertion,
then inserting $F$ can be handled without compromising the running time 
albeit at the cost of introducing additional error.
\newcommand{\roff}{\rho}
\begin{restatable}{lemma}{larbitraryextension}\label{l:arbitrary-extension}
The data structure of~\Cref{t:source-sssp} can be extended to also support
inserting batches $F\subseteq V\times V$ of arbitrary weighted edges to $G$ if provided with an integer exponent $\alpha\geq 0$ such that
the currently stored estimates satisfy
$d(v)\leq (1+\epsp)^{\alpha}\dist_{G+F}(s,v)$ for all $v\in V$.

Then, at all times the maintained estimates satisfy:
\[ d(v)\leq (1+\epsp)^{1+\roff+\log_2{m}} \dist_G(s,v)\text{ for all }v\in V, \]
where $\roff$, called the \emph{rank offset}, equals the exponent $\alpha$ of the most recently performed batch-insertion. The total
update time of the data structure remains $\O(m\log(n)\log(nW)/\epsp+\Delta)$.
\end{restatable}
One can view the basic data structure of~\Cref{t:source-sssp} (supporting only source insertions) as keeping the rank offset $\roff$ defined as above always zero.

To prove~\Cref{l:arbitrary-extension}, we start with the following lemma.
\begin{lemma}\label{l:insert-arbitrary}
  Let $F\subseteq V\times V$ be some set of weighted edges.
  Suppose for some integer $\alpha\geq 0$, the maintained estimates $d:V\to\{0\}\cup [1,nW]$ for $G$ additionally satisfy
  \begin{equation}\label{eq:guarantee}
    d(v)\leq (1+\epsp)^{\alpha}\dist_{G+F}(s,v)
  \end{equation}
  for all $v\in V$. Then, after inserting the edges $F$ one by one according to~\Cref{obs:insert},~\Cref{slack-inv} is satisfied and
  every vertex is $\alpha$-certified.
\end{lemma}
\begin{proof}
  Note that inserting the edges according to~\Cref{obs:insert} can only decrease estimates $d(\cdot)$.
  Thus, if the inequality~\eqref{eq:guarantee} holds before the insertion, it will also hold afterwards.
  By construction, after processing the insertion, for all $v\in V$, $d(v)$ represents
  the length of some $s\to v$ path in $G+F$, and thus $\dist_{G+F}(s,v)\leq d(v)$.

  Consider any path $P=u\to v$ in $(G+F)[V\setminus \{s\}]$. We have:
  \begin{align*}
    d(v)&\leq (1+\epsp)^{\alpha}\cdot \dist_{G+F}(s,v)\\
        &\leq (1+\epsp)^{\alpha} \cdot (\dist_{G+F}(s,u)+\wei(P))\\
        &\leq (1+\epsp)^{\alpha} \cdot (d(u)+\wei(P)).
  \end{align*}
  Since the path $P$ was chosen arbitrarily, $u$ is $\alpha$-certified after applying the said operation.
\end{proof}
Applying \Cref{l:insert-arbitrary} with $\alpha\geq 1$ cannot guarantee that~\Cref{l:rank} holds after inserting a batch~$F$,
as some of the vertices might have had their rank smaller than $\alpha$ before.
As a result, their ranks might be incorrect now.
To deal with this, we slightly relax the invariant posed on the maintained ranks. Specifically, we augment the data structure of~\Cref{t:source-sssp} to also maintain an integral \emph{rank offset} $\roff\geq 0$
such that \emph{every vertex $v\in V$ is $(\roff+\rank(v))$-certified}.
One can view the basic data structure of~\Cref{t:source-sssp} as keeping the rank offset defined like this always zero.
This way, whenever~\Cref{l:insert-arbitrary} is applied given some exponent $\alpha$,
we can simply reset $\roff:=\alpha$.
It is easy to verify that both~\Cref{sync_correct}~and~\Cref{l_is_log} remain valid after this change.
However, the bound in~\Cref{l:rank} gets weakened to
\begin{equation}
  d(t)\leq (1+\epsp)^{\roff+r+1}\cdot\dist(s,t)\leq (1+\epsp)^{\roff+\log_2{m}+1}\cdot \dist(s,t).
\end{equation}
The running time analysis of the data structure of~\Cref{t:source-sssp} goes through unaltered.
Therefore, using~\Cref{l:insert-arbitrary}, we obtain~\Cref{l:arbitrary-extension}.

\subsection{Resetting the data structure}\label{sec:resetting}
Using~\Cref{l:arbitrary-extension} repeatedly may lead to a significant growth of the rank offset
which makes the vertex certificates larger and thus the maintained estimates less and less accurate.
Later on, in our incremental APSP application (\Cref{sec:all-pairs}), this will indeed prove to be a problem
and we will need to keep the rank offset bounded.
In this section we discuss some ways to achieve that.
\paragraph{A deterministic recompute-from-scratch reset.}
Recall from \Cref{lem:no_slack} that running \linebreak $\PD(V)$ makes all vertices $0$-certified.
As a result, by running $\PD(V)$ and subsequently setting all vertex ranks
and the rank offset $\roff$ to $0$, in $\O(m+n\log{n})$ time we can get rid of
the pathwise error of the maintained estimates $d(\cdot)$.

Note that in the running time analysis so far, the work performed
inside $\textsc{Synchronize}$ was fully charged to the unit rank increases that made the ranks of the $\PD$'s input vertices positive.
Consequently, as the discussed reset procedure zeroes out ranks, invoking it any number of times interleaved with
the other discussed operations does not increase
the total update time bound of the other operations of the data structure.

\paragraph{Randomized reset.}
Reducing the vertex certificates all the way down to $0$ might be unnecessary; keeping the certificates
bounded might be just enough.
In~\Cref{a:reset}, we discuss the randomized approach of~\cite{ChechikZ21} used to achieve that (adapted to our needs).
Formally, we show the following:

\newcommand{\bal}{\lambda}

\begin{restatable}{theorem}{trandomizedreset}\label{t:randomized-reset}
Let $\bal\in [1,\ell]$ be an integer. Then, in $\O\left(m\log^2(nW)\log^2(n)/(\bal^3\epsp^2)\right)$ additional time
one can adjust the estimates so that every vertex is $\lambda$-certified with high probability.
\end{restatable}

\section{Incremental all-pairs shortest paths for adaptive adversaries}\label{sec:all-pairs}
In this section, we describe our incremental APSP data structure. We prove:
\tallpairs*

\newcommand{\Vcur}{V_{\mathrm{cur}}}
\newcommand{\Ecur}{E_{\mathrm{cur}}}
\newcommand{\Gbeg}{G_{\mathrm{beg}}}
\newcommand{\batchds}{\mathcal{B}}
\newcommand{\dsfrom}{\mathcal{D}}
\newcommand{\dsto}{\mathcal{D}^{R}}
\newcommand{\apsp}{\mathcal{A}}
\newcommand{\Gto}{G^{\mathrm{to}}}
\newcommand{\Gfrom}{G^{\mathrm{from}}}
\newcommand{\phlen}{b}
\newcommand{\reset}{p}

Let $\phlen\in [1,m]$ be an integer parameter to be set later.
The data structure operates in phases of $\phlen$ edge updates.
While a phase proceeds, denote by $\Ecur$ the edges whose insertions or weight decreases have been issued in the current phase.
Let $\Vcur$ be the set of endpoints of $\Ecur$, so that $|\Vcur|\leq 2b$. Moreover, let $\Gbeg$ denote the graph $G$ from the beginning of
the phase, that is, with all insertions from the previous phases applied.
Note that $G=\Gbeg+\Ecur$ at all times.

Let $\reset\in \mathbb{Z}_+$ be a \emph{reset parameter} and $\epsp\in (0,1)$ be an \emph{accuracy parameter},
both also to be set later.
Roughly speaking, the accuracy of the data structure will deteriorate slightly after each phase,
and will be restored every $\reset$ phases.
The accuracy parameter will be shared by all the data structure's components and will
be adjusted based on $\eps$ and other parameters at the very end.

We will now describe the data structure's components and the interplay between them.
For convenience, we will use the following notation: if $\mathcal{C}$ is a component maintaining
some vertex estimates, we will write $\mathcal{C}(v)$ to denote the estimate that $\mathcal{C}$ maintains for
the vertex $v$.

Before we continue, let us refer to a standard way of reducing to the case when there is only at most $\O(m\log(W)/\eps)$ weight decreases in total (see, e.g., \cite{Bernstein16}).
Indeed, we may only record a weight decrease for some edge $e\in E$ if the new weight of $e$ is smaller than the previously recorded weight of $e$ by a factor of at least $1+\eps$, and ignore it otherwise.
Having done that, by proceeding normally afterwards, the returned distance estimates might be distorted by an extra $1+\eps$ factor (in addition to the $1+\eps$ factor we aim for).
But this can be easily dealt with at no asymptotic cost by decreasing $\eps$ by a constant factor (say $4$) in the very beginning.
Filtering weight decreases like that clearly costs only $\O(\Delta)$ additional time.
Note that with this updates filtering scheme applied, we can bound the number of phases by $\O(m\log(W)/(\phlen\eps))$.

\paragraph{SSSP data structures with shortcuts.}
For each $s\in V$, we maintain a data structure $\dsfrom_s$ of~\Cref{t:source-sssp} extended to support
arbitrary batch insertions as described in~\Cref{l:arbitrary-extension}.
The data structure $\dsfrom_s$ maintains a graph $G_s$ obtained from $\Gbeg$ by extending
it with $n$ \emph{shortcut edges} $e_{s,u}=su$.
To the best of our knowledge, the idea of using shortcut edges for partially dynamic APSP is due to~\cite{Bernstein16} and has been applied in the incremental data structure of~\cite{KarczmarzL19} as well.

The weight of a shortcut edge $e_{s,u}$
always corresponds to the length of some $s\to u$ walk in \emph{the current graph} $G$.
Hence, $\wei_{G_s}(e_{s,u})\geq \dist_G(s,u)$.
As a result, for each $v\in V$, we have
\[ \dist_{G}(s,v)\leq \dist_{G_s}(s,v)\leq \dist_{\Gbeg}(s,v). \]
Our goal will be to maintain the shortcuts so that after each update, the following holds:
\[ \dist_{G_s}(s,v)\leq (1+\epsp)^{\O(p\log{n})}\cdot \dist_G(s,v). \]
For an appropriate $\epsp$, we will set $d(u,v):=\dsfrom_u(v)$ to be the final outputs of our data structure.

Recall that using the extension (\Cref{l:arbitrary-extension}) requires taking into account the rank offset $\roff$ defined therein.
All our single-source data structures will use the same rank offset $\roff$ which will grow slightly after every phase,
but at the same time it will be kept under control using resets.

We also maintain a symmetric collection of data structures of~\Cref{t:source-sssp} on the reverse
graph~$\rev{G}$, i.e., $G$ with edge directions reversed. That is, a data structure $\dsto_t$ maintains $\rev{\Gbeg}$ with some shortcuts $\rev{e_{t,u}}$
corresponding to paths from~$t$ in $\rev{G}$ added.
The actual purpose of $\dsto_t$ is to maintain approximate shortest paths \emph{to the vertex $t$ in $G$ from all $s\in V$};
running a data structure on the reverse graph turns such a single-sink problem into a single-source problem.

The single-source data structures $\dsfrom_s$ and $\dsto_t$ are shared by all the phases.

\paragraph{APSP between insertion endpoints.} Fix some phase.
Let $\Ecur(u,v)$ denote the smallest weight of an edge $uv\in \Ecur$, or $\infty$ if no such edge exists in $\Ecur$.
Let $H$ be a complete digraph on~$\Vcur$ such that
for all $u,v\in \Vcur$, we have $\wei_H(uv):=\min\left(\dsfrom_u(v),\Ecur(u,v)\right)$.

Note that while the phase proceeds, $\Vcur$ grows and thus the graph $H$ also grows.
Moreover, the weights of $H$ undergo weight decreases while the estimates
in the data structures $\dsfrom_u$, where $u\in \Vcur$, change.
We store the graph $H$ using the near-optimal incremental APSP data structure $\apsp$ for dense graphs, formally characterized below.

\begin{theorem}\label{t:dense-apsp}{\upshape{\cite{KarczmarzL20}}}
  Let $\epsp\in (0,1)$. Let $G=(V,E)$ be a directed graph with edge weights in $\{0\}\cup [1,W]$.
There exists a deterministic data structure maintaining $G$ subject to edge insertions
and weight decreases explicitly maintaining for all pairs $u,v\in V$ an estimate $d(u,v)$ such that
\begin{itemize}
  \item $d(u,v)$ is the weight of some $u\to v$ path in $G$, ie., $\dist_G(u,v)\leq d(u,v)$,
  \item $d(u,v) \leq (1+\epsp)\dist_G(u,v)$.
\end{itemize}
The total update time is $\O(n^3\log(n)\log(nW)/\epsp+\Delta)$,
where $\Delta$ is the total number of updates issued.
\end{theorem}
\begin{remark}
The above data structure~{\upshape{\cite[Lemma~5.4]{KarczmarzL20}}}, stated originally, does not enforce the returned estimates to be lengths of actual paths in $G$.
However, it can be easily modified to maintain ``witnesses'', i.e., weights of paths $u\to v$ of weight smaller than the corresponding estimate $d(u,v)$.
\end{remark}
\begin{remark}
The data structure in~\Cref{t:dense-apsp} assumes a fixed vertex set $V$, whereas in our use case, the vertex set grows. However, we know a bound $2\phlen$ on the final size of $\Vcur$ so we can start the data structure of~\Cref{t:dense-apsp} with a vertex set of size $2\phlen$ and map the remaining placeholder vertices to new vertices entering $\Vcur$ while the phase proceeds.
\end{remark}
The data structure~$\apsp$ is reinitialized from scratch each time a new phase starts.
In the following, we will write $\apsp(u,v)$ to denote the estimate $d(u,v)$ maintained by the data structure $\apsp$.

\paragraph{Setting the shortcuts' weights.}
We are now ready to describe how the weights of the shortcuts in the data structures $\dsfrom_s$ and
$\dsto_s$ are defined and updated.

We maintain the following invariants:
\begin{enumerate}[label=(\arabic*)]
\item $\wei(e_{s,t})\leq \apsp(s,t)$ and $\wei(\rev{e_{t,s}})\leq \apsp(s,t)$ for every $s,t\in \Vcur$,
\item $\wei(e_{s,t})\leq \dsto_t(s)$ for every $s\in V$ and $t\in \Vcur$.
\item $\wei(\rev{e_{t,s}})\leq \dsfrom_s(t)$ for every $s\in \Vcur$ and $t\in V$.
\end{enumerate}
Let us explain how invariant~(1) above is maintained; the other are maintained analogously.
Whenever the data structure $\apsp$ changes some of its maintained estimates $\apsp(u,v)$, the change is passed
as a source-weight-decrease to the data structures $\dsfrom_u$ and $\dsto_v$.
The existing weight of the relevant shortcut therein might be already smaller; in such a case
the weight decrease is simply ignored.
\paragraph{Closing a phase.}
When a phase is completed, the edges $\Ecur$ are batch-inserted to $\dsfrom_s$ and $\dsto_s$
for all $s\in V$ using Lemma~\ref{l:arbitrary-extension}. For this to be allowed, 
estimates in these data structures (storing the graphs $\Gbeg$ augmented with some shortcuts) need to approximate the distances
in $G$ well.
We will later prove (\Cref{l:apsp-accuracy}) that this is the case indeed.
Moreover, we will show that the rank offset of the single-source data structures
will only grow by $\O(\log{m})$ as a result of these batch-insertions.

\paragraph{Resets.} Every $\reset$ phases, we perform a reset in all the maintained single-source
data structures, as described in Section~\ref{sec:resetting}.
In the deterministic variant of our data structure, we use the deterministic reset,
which reduces the parameter $\roff$ in each of them all the way to $0$.
In the randomized variant, we use randomized resetting (\Cref{t:randomized-reset}) to reduce the rank offset $\roff$ to
$\roff^*:=\lceil \reset\log_2{m}\rceil$.

\subsection{Estimate error analysis}
Let $k< \reset$ denote the number of phases completed since the last reset. 
In this section, we
bound the errors of the individual components of the data structure
as a function of $k$.

Let $y:=\log_2{m}+1$.
The single-source data structures $\dsfrom_s$ and $\dsto_s$ (for $s\in V$) will all use the rank
offset $\roff$ no more than $\roff^*+k\cdot 4y$.
Note that for $k=0$, that is, immediately after a reset, this is indeed achieved by construction.
The below lemma bounding the distortions of the estimates maintained in individual components
is the key to the correctness.

\begin{restatable}{lemma}{lapspaccuracy}\label{l:apsp-accuracy}
Let $k< \reset$ denote the number of phases completed since the last reset. Then:
  \begin{enumerate}[label=(\roman*)]
\item for all $s,t\in V$, $\dsfrom_s(t)\leq (1+\epsp)^{\roff^*+k\cdot 4y}\cdot \dist_{\Gbeg}(s,t)$,
\item for all $s,t\in V$, $\dsto_t(s)\leq (1+\epsp)^{\roff^*+k\cdot 4y}\cdot \dist_{\Gbeg}(s,t)$,
\item for all $s,t\in \Vcur$, $\apsp(s,t)\leq (1+\epsp)^{\roff^*+k\cdot 4y+y}\cdot \dist_{G}(s,t)$,
\item for all $s\in \Vcur$ and $t\in V$, $\dsfrom_s(t)\leq (1+\epsp)^{\roff^*+k\cdot 4y+2y}\cdot \dist_G(s,t)$,
\item for all $s\in V$ and $t\in \Vcur$, $\dsto_t(s)\leq (1+\epsp)^{\roff^*+k\cdot 4y+2y}\cdot \dist_G(s,t)$,
\item for all $s,t\in V$, $\dsfrom_s(t)\leq (1+\epsp)^{\roff^*+k\cdot 4y+3y}\cdot \dist_G(s,t)$,
\item for all $s,t\in V$, $\dsto_t(s)\leq (1+\epsp)^{\roff^*+k\cdot 4y+3y}\cdot \dist_G(s,t)$.
\end{enumerate}
\end{restatable}
\begin{proof}
Consider item~(i). For $k=0$ it follows by the guarantees of the reset. That is, when a deterministic
reset is performed at the beginning of the phase, we have $\dsto_s(t)\leq \dist_{\Gbeg}(s,t)$
  and if the reset is randomized, we have $\dsto_s(t)\leq (1+\epsp)^{\roff^*}\cdot \dist_{\Gbeg}(s,t)$.
Suppose $k\geq 1$. 
Then, just before the edges from the last phase were inserted into $\dsfrom_s$, we had
  \[\dsfrom_s(t)\leq (1+\epsp)^{\roff^*+(k-1)\cdot 4y+3y}\cdot\dist_G(s,t)\] by item~(vi) applied
for $(k-1)$.
Thus, by~\Cref{l:arbitrary-extension}, after the insertion we have
  \[ \dsfrom_s(t)\leq (1+\epsp)^{\roff^*+(k-1)\cdot 4y+3y+y}\cdot\dist_G(s,t)=(1+\epsp)^{\roff^*+k\cdot 4y}\cdot \dist_G(s,t). \]
The proof of item~(ii) is completely symmetric.

To prove~(iii), consider the graph $H$ maintained by $\apsp$. Let $P$ be the shortest $s\to t$ path in $G$. Express $P$ as $P_1e_1P_2e_2\ldots e_{k-1}P_k$ where
$k\geq 1$, $u_i\to v_i=P_i\subseteq \Gbeg$ is a possibly empty path and $e_i\in \Ecur$.
Note that $u_i,v_i\in \Vcur$ since $s,t\in \Vcur$ and all endpoints of the edges $e_i$ are in $\Vcur$
by definition.
Moreover, by construction we have $\dist_G(u_i,v_i)=\dist_{\Gbeg}(u_iv_i)$ for all $i$
and $\dist_G(v_i,u_{i+1})=\wei_G(e_i)$ for all $i<k$.
By item~(i), for all $i$ we have
\[ \wei_H(u_i,v_i)\leq \dsto_{u_i}(v_i)\leq (1+\epsp)^{\roff^*+k\cdot 4y}\cdot \dist_{\Gbeg}(s,t)=(1+\epsp)^{\roff^*+k\cdot 4y}\cdot \dist_{G}(u_i,v_i). \]
Similarly, we get:
\[ \wei_H(v_iu_{i+1})\leq \wei_G(e_i) = \dist_G(v_iu_{i+1}).\]
It follows that there is a path $s=u_1\to v_1\to u_2\to v_2\to \ldots \to u_k\to v_k=t$ in $H$ of weight at most $(1+\epsp)^{\roff^*+k\cdot 4y}\dist_G(s,t)$.
Hence, by~\Cref{t:dense-apsp},
\[\apsp(s,t)\leq (1+\epsp)\dist_H(s,t)\leq (1+\epsp)^{\roff^*+k\cdot 4y+y}\dist_G(s,t).\]

Now consider items~(iv)~and~(v). We will only prove~(iv), as the proof of~(v) is completely symmetric.
Consider a shortest path $P=s\to t$ in $G$. Express $P$ as $P_1P_2$, where $P_2=u\to t$ is such that $u$ is the last vertex from $\Vcur$ on $P$ (possibly $u=s$ or $u=t$).
Note that $E(P_2)\cap \Ecur=\emptyset$, so $P_2\subseteq \Gbeg$.
Since $s,u\in \Vcur$, by item~(iii) and the invariant~(1) posed on the shortcut edge $e_{s,u}$,
we have:
\[ \wei_{G_s}(e_{s,u})\leq \apsp(s,u)\leq (1+\epsp)^{\roff^*+k\cdot 4y+y}\cdot \wei(P_1).\]
But $P_2\subseteq \Gbeg\subseteq G_s$, so the $s\to t$ path $e_{s,u}\cdot P_2$ in $G_s$ has weight at most $(1+\epsp)^{\roff^*+k\cdot 4y+y}\cdot \dist_G(s,t)$.
Since $\dsfrom_{s}(t)\leq (1+\epsp)^y\cdot \dist_{G_s}(s,t)$, we obtain the desired inequality
\[ \dsfrom_s(t)\leq (1+\epsp)^{\roff^*+k\cdot 4y+2y}\cdot \dist_G(s,t). \]

Finally, we discuss items~(vi)~and~(vii). We only prove (vi), since the proof of (vii) is completely symmetric.
Consider  a shortest path $P=s\to t$ in $G$.
If $P\subseteq \Gbeg$, then $P\subseteq G_s$, so ${\dsfrom_s(t)\leq (1+\epsp)^y\dist_G(s,t)}$ by~\Cref{t:source-sssp}.
Otherwise, express $P$ as $P_1P_2$, where $P_2=u\to t$ is such that $u$ is the last vertex from $\Vcur$ appearing on $P$.
Note that $P_2\subseteq \Gbeg\subseteq G_s$, so $\dist_{G}(u,t)=\dist_{G_s}(u,t)$.
By invariant~(2) posed on the shortcut edge $e_{s,u}$ and item~(v), we get:
\begin{align*}
  \dist_{G_s}(s,t)&\leq \wei(e_{s,u})+\dist_{G_s}(u,t)\\
                  &\leq \dsto_u(s)+\dist_G(u,t)\\
                  &\leq (1+\epsp)^{\roff^*+k\cdot 4y+2y}\cdot \dist_G(s,u)+\dist_G(u,t)\\
                  &\leq (1+\epsp)^{\roff^*+k\cdot 4y+2y}\cdot \dist_G(s,t).
\end{align*}
It follows that $\dsfrom_s(t)\leq (1+\epsp)^{\roff^*+k\cdot 4y+3y}\cdot \dist_G(s,t)$, as desired.
\end{proof}

By items~(vi)~and~(vii) of~\Cref{l:apsp-accuracy} and since $\roff^*+k\cdot 4y+4y=\roff^*+(k+1)\cdot 4y$, at the end of the phase one can indeed
insert the edges $\Ecur$ to all the data structures $(\dsfrom_s)_{s\in V}$ and $(\dsto_t)_{t\in V}$
using~\Cref{l:arbitrary-extension}, set $\roff:=\roff^*+(k+1)\cdot 4y$, and initialize the next phase.

By~\Cref{l:apsp-accuracy}, we conclude that that the estimates stored in data structures $(\dsfrom_s)_{s\in V}$
satisfy:
\[ \dsfrom_s(t)\leq (1+\epsp)^{\roff^*+4\reset\cdot y}\cdot \dist_G(s,t). \leq (1+\epsp)^{5\reset\cdot y}\cdot \dist_G(s,t).\]
Thus, as long as $\epsp\leq \frac{\eps}{10\reset y}$, they indeed constitute $(1+\eps)$-approximate distance estimates in $G$.

\subsection{Running time analysis}
First of all, by~\Cref{t:source-sssp}~and~\Cref{l:arbitrary-extension},
the total cost of setting up and maintaining the single-source data structures $(\dsfrom_s)_{s\in V}$
and $(\dsto_t)_{t\in V}$ through all updates, excluding the resets and the $\Delta$ terms denoting the numbers of updates issued to these data structures is $O(n\cdot m\log(nW)\log(n)/\epsp)$. 

For each of the $\O(m\log(W)/(\eps\phlen))$ phases, we reinitialize and run a fresh data structure of~\Cref{t:dense-apsp}
on a graph with $\O(\phlen)$ vertices and $\O(\phlen^2)$ edges.
The total update time of that data structure, excluding the $\Delta$ term counting the issued updates, is $\O(\phlen^3\log(nW)\log(n)/\epsp)$ per phase.
Hence, the total cost incurred through all phases is $\O(m\phlen^2\log^2(nW)\log(n)/(\eps\epsp))$.

Now, let us consider the number of updates issued to the maintained components.
Recall our assumption that the total number of updates issued to $G$ is $\O(m\log(nW)/\eps)$.
Consequently, the total number of ``non-shortcut'' updates issued to the single-source data structures is $\Ot(nm\log(nW)/\eps)$.
Similarly, the number of updates issued to the data structures of~\Cref{t:dense-apsp} that are not caused by some estimate change in the single-source data structures is $\O(m\log(nW)/\eps)$.

Every update of a ``shortcut'' edge in the maintained data structures is caused by an estimate change in another component.
Hence, the total number of shortcut edge updates can be bounded by the total number of estimate changes in all the maintained components. 
That quantity can be bounded in turn by the total update time of all the maintained components ignoring the $\Delta$ terms denoting the numbers of
updates issued.
We conclude that this additional cost is asymptotically no more than what we have taken into account so far.

Recall that in the deterministic~variant, a reset costs $\O(mn\log{n})$ time.
In the randomized variant, by~\Cref{t:randomized-reset} applied with~${\bal=\roff^*}$, a reset costs $\O\left(\frac{nm\log^2(nW)}{\reset^3\epsp^2\log{n}}\right)$ time.
Recall that the number of resets is $\O(m\log(W)/(\eps\phlen\reset))$,
so the total cost of all resets is $\O\left(m^2n\log(W)\log(n)/(\eps\phlen\reset)\right)$
in the deterministic variant and $\O\left(\frac{m^2n\log^3(nW)}{\eps\phlen\reset^4\epsp^2\log{n}}\right)$
in the randomized variant.
By using \linebreak ${\epsp:=\eps/10\reset(\log{m}+1)}$, the final running time in the deterministic case becomes:
\[ \O\left(nm\log(nW)/\eps+nm\reset\log(nW)\log^2(n)/\eps+m\phlen^2\reset\log^2(nW)\log^2(n)/\eps^2+\frac{m^2n\log(n)\log(W)}{\eps\cdot \phlen\cdot \reset}\right). \]
By setting $\phlen=\sqrt{n}$ and $\reset=\frac{m^{1/2}\eps^{1/2}}{n^{1/4}\log{n}}$, the running
time becomes $\O(m^{3/2}n^{3/4}\log^2(nW)\log(n)/\eps^{3/2})$.

In the randomized case, the last term is replaced with
$m^2n\log^3{(nW)}\log{n}/(\phlen \reset^2 \eps^3)$.
By setting $\phlen=\sqrt{n}$ and $p=\frac{m^{1/3}}{n^{1/6}\eps^{1/3}}$, the running time becomes
$\O(m^{4/3}n^{5/6}\log^3(nW)\log(n)/\eps^{7/3})$.
\section{Offline incremental SSSP in near-linear time}\label{sec:offline}
In this section we describe a simple near-optimal offline incremental $(1 + \epsilon)$-approximate SSSP data structure and prove the following:
\restateOfflineTheorem*

\newcommand{\Res}{\mathcal{D}}
\newcommand{\dla}{d_{\le\alpha}}
\newcommand{\dgb}{d_{\ge\beta}}
\newcommand{\dg}{d_{\gamma}}
\newcommand{\dlg}{d_{\le\gamma}}
\newcommand{\dgg}{d_{\ge\gamma}}
\newcommand{\dinit}{d_\mathrm{init}}
\newcommand{\VV}{X}

\paragraph{Main idea.}

We first give a quick overview of the approach.
We construct a collection $\Res$ of distance estimates found for different vertices and different versions of the graph.
Assuming that for some query $(u, j)$ we store an estimate of $\dist_{G^j}(s, u)$ in $\Res$, we will be able to answer that query directly.
If the answer cannot be retrieved from $\Res$ directly, we look for the smallest estimate of $\dist_{G^i}(s, u)$ stored for some $i < j$ that we can find in $\Res$ and return that value as the answer instead.
If we ensure that for each such $j$
we can also find in $\Res$ an estimate of $\dist_{G^k}(s, u)$, where $k > j$
that is ``close enough'' to the reported estimate of $\dist_{G^i}(s, u)$, we can use the two ``surrounding'' estimates as lower and upper bounds for the distance sought in the query.

We build the collection $\Res$ outlined above as follows.
We first find exact distances in the initial and final versions $G^0$ and $G^\Delta$ and store them in $\Res$.
We then use a divide-and-conquer approach.
We use a recursive $\textproc{Search}$ subroutine (\Cref{alg:offline}), which given the current set $\Res$ and a range $[\alpha, \beta]$, determines the set of vertices $u$ such that can estimate $\dist_{G^j}(u)$ for all $j \in [\alpha, \beta]$ using the estimates already stored, and searches for additional estimates for the remaining vertices.
\paragraph{Setup.}

Similarly as with the previous data structures, let us assume for simplicity that $G$ initially contains $n$ edges $su$ of weight $nW$, so that $\dist(s, u) \le nW$ for all vertices $u$ throughout all updates.
Let $\epsp$ be an accuracy parameter, set to $\eps / (2 \log_2(\Delta) + 2)$.
We will write $\dla(u)$ to denote the smallest estimate of $\dist_{G^j}(s,u)$ stored in $\Res$ for some $j \le \alpha$
and similarly write $\dgb(u)$ to denote the largest estimate of $\dist_{G^j}(s,u)$ stored in $\Res$ for some $j \ge \beta$.
When using this notation we consider all the estimates that are \emph{currently} stored in $\Res$; what ``currently'' means depends on the context. 

\subsection{The algorithm}
We start by storing in $\Res$ the exact value of $\dist_{G^0}(s, u)$ and $\dist_{G^\Delta}(s, u)$ for each vertex $u$.
We find them by running Dijkstra's algorithm on $G^0$ and $G^\Delta$.
We then call the procedure $\textproc{Search}(0, \Delta)$, which accesses $\Res$ and stores some additional distance estimates of $\dist_{G^j}(s,u)$ for different $j$ and $u$ (see \Cref{alg:offline}).
After that, the construction of $\Res$ is complete and we can process the queries.

For each query concerning $\dist_{G^j}(s,u)$ we return the value $d_{\le j}(u)$, which corresponds to the smallest estimate of $\dist_{G^i}(s,u)$ for some $i \le j$ stored in $\Res$.

\begin{algorithm}[t]
\caption{\textproc{Search}($\alpha$, $\beta$)}\label{alg:offline}
Let $\Res$ be a global set that stores some estimates of $\dist_{G^j}(s,u)$ for different $u$ and $j$. \\
Let $\dla(u)$ denote the smallest estimate of $\dist_{G^j}(s,u)$ stored in $\Res$ for any $j \le \alpha$. \\
Let $\dgb(u)$ denote the largest estimate of $\dist_{G^j}(s,u)$ stored in $\Res$ for any $j \ge \beta$.
\begin{algorithmic}[1]
  \State If $[\alpha, \beta]$ is empty, end the procedure
  \State $\VV := \set{u\in V \ : \  \dla(u) > (1 + \epsp)\dgb(u)}$\label{line:VV}
  \State $\gamma := \left\lfloor\left(\alpha + \beta\right) / 2\right\rfloor$
  \State Construct a graph $H$ equal to $G^\gamma[\VV]$ with an additional source vertex $f$
  \For{$v \in \VV$}
        \State $\dinit(v) := \min\set{\dla(u) + w(e) \ : \ uv=e \in E(G^\gamma), \ u \not \in \VV}$ \label{line:d_init}
    \State In $H$ connect $f$ to $v$ with an edge of weight $\dinit(v)$
  \EndFor
  \State $\dg := \Call{Dijkstra}{H, f}$ \Comment{$\dg(v) = \dist_H(f, v)$ for all $v$} \label{line:Dijkstra}
  \For{$v \in \VV$}
    \State Store $d_\gamma(v)$ in $\Res$ as an approximation of $\dist_{G^\gamma}(s,v)$
  \EndFor
  \State \Call{Search}{$\alpha$, $\gamma - 1$}
  \State \Call{Search}{$\gamma + 1$, $\beta$}
\end{algorithmic}
\end{algorithm}

\subsection{Correctness analysis}
In this section we analyze the error buildup of the estimates stored in $\Res$ by different $\textproc{Search}$ calls, depending on their depth in the recursion tree, where we assume that the depth of the $\textproc{Search}(0, \Delta)$ call that we invoke directly is $0$. 
We then bound the error of the reported answers.

\begin{restatable}[]{lemma}{restateOfflineCorrectness}\label{theorem:error bound}
Consider a call $\textproc{Search}(\alpha, \beta)$ at depth $h$. Then for every $t\in \VV$, we have
\[ \dist_{G^\gamma}(s,t)\leq d_\gamma(t) \leq (1+\epsp)^{h+1}\cdot \dist_{G^\gamma}(s,t). \]
\end{restatable}
\begin{proof}
The proof is by induction on $h$. Assume that the statement holds for all depths smaller than~$h$.
First observe that for every vertex $u \not\in \VV$, the current value $\dgb(u)$ is a  $(1 + \epsp)^{h}$\nobreakdash-approximation of $\dist_{G^j}(s,u)$ for some $j \ge \beta$.
Indeed, if $\dgb(u)$ was inserted to $\Res$ during initialization, it is a $(1 + \epsp)^0$-approximation of $\dist_{G^\Delta}(s,u)$.
Otherwise, it was inserted by some ancestor call at depth $h'$ smaller than $h$.
Therefore, by the induction hypothesis we have 
\begin{equation}
  \dla(u) \le (1 + \epsp) \dgb(u) \le (1 + \epsp)^{h + 1} \dist_{G^j}(s,u) \le (1 + \epsp)^{h + 1}\dist_{G^\gamma}(s,u), \label{outside VV equation}
\end{equation}
and since $\dla(u) \ge \dist_{G^\gamma}(s,u)$, indeed $\dla(u)$ is a $(1 + \epsp)^{h + 1}$-approximation of $\dist_{G^\gamma}(s,u)$.

Now let $t \in \VV$ and take any simple shortest path $P_{st}$ from $s$ to $t$ in $G^\gamma$.
Let $e=uv$ be the last edge on $P_{st}$,  such that $u \not\in \VV$ and $v \in \VV$, the existence of which follows from $s\notin \VV$.
By definition, $\dinit(v) \le \dla(u) + w(e)$.
Let $P_{vt}$ denote the suffix of the path $P_{st}$ going from vertex $v$ to $t$.
The length of $P_{vt}$ is equal to $\dist_{G^\gamma}(v, t)$ as it must itself be a shortest path.
Observe that since $P_{vt}$ consists of vertices from the set $\VV$, it exists in the graph $H$, implying $\dist_H(v, t) \le \dist_{G^\gamma}(v, t)$.
Putting everything together we get
\eq{
  \dist_H(f, t) &\le \dist_H(f, v) + \dist_H(v, t)\\ 
  &\le \dinit(v) + \dist_{G^\gamma}(v, t) \\
  &\le \dla(u) + w(e) + \dist_{G^\gamma}(v, t) \\
  &\le (1 + \epsp)^{h + 1}\dist_{G^\gamma}(s, u) + w(e) + \dist_{G^\gamma}(v, t) \\
  &\le (1 + \epsp)^{h + 1}\left(\dist_{G^\gamma}\left(s, u\right) + w\left(e\right) + \dist_{G^\gamma}\left(v, t\right)\right) \\
  &= (1 + \epsp)^{h + 1}\dist_{G^\gamma}(s, t).
}

The lower bound $\dist_H(f, t) \ge \dist_{G^\gamma}(s, t)$ can be proven similarly.
Let $P_{ft}$ be any simple shortest path from $f$ to $t$ in $H$ and let $v$ be vertex following $f$ on $P_{ft}$.
By definition we have $\dist_H(f, v) = \dinit(v) = \dla(u) + w(e)$ for some $e = uv \in G^\gamma$.
We know that the remainder of the path $P_{ft}$ (from $v$ onwards) is contained in $G^\gamma$, thus $\dist_{G^\gamma}(v, t) \le \dist_H(v, t)$.
We thus get
\eq{
  \dist_H(f, t) &= \dist_H(f, v) + \dist_H(v, t) \\
  &\ge \dla(u) + w(e) + \dist_{G^\gamma}(v, t) \\
  &\ge \dist_{G^\gamma}(s, u) + w(e) + \dist_{G^\gamma}(v, t) \\
  &\ge \dist_{G^\gamma}(s, t),
}
which completes the proof.
\end{proof}

Now observe that the depth of every (not immediately ending) recursive call is at most $\log_2(\Delta)$ since the interval $[\alpha,\beta]$ is halved in children calls.
By putting $\epsp:=\eps/(2\log_2(\Delta)+2)$, the approximation factor for every estimate stored in $\Res$ is $1 + \eps$ by reasoning similar to \Cref{remark: epsp fix}.

\begin{restatable}[]{corollary}{restateOfflineCorollary}
For any query $(u, j)$, $d_{\le j}(u)$ is a $(1 + \eps)$-approximate estimate of $\dist_{G^j}(s,u)$.
\end{restatable}
\begin{proof}
Let $\textproc{Search}(\alpha, \beta)$ be the call for which $\gamma = \left\lfloor\left(\alpha + \beta\right) / 2\right\rfloor$ is equal to $j$ -- observe that there must be exactly one such call in the recursion tree.
If $u \in \VV$ for that call, we know that $\dg(u)$ is stored in $\Res$ and by \Cref{theorem:error bound} we get
$d_{\le j}(u) \le (1 + \epsp)^{h + 1}\dist_{G^j}(s,u)$, where $h$ denotes the depth of the call.
If $u \not \in \VV$, we had $\dla(u) \le (1 + \epsp)\dgb(u)$ at the beginning of the call, implying that the answer is also at most
$(1 + \epsp)^{h + 1}\dist_{G^j}(s,u)$ analogously as in \eqref{outside VV equation} from the proof of \Cref{theorem:error bound}.
We get the lower bound $d_{\le j}(u) \ge \dist_{G^j}(s,u)$ simply because $d_{\le j}(u) \ge \dist_{G^i}(s,u)$ for some $i \le j$ and of course $\dist_{G^i}(s,u) \ge \dist_{G^j}(s,u)$.
Since $(1 + \epsp)^{h + 1} \le 1 + \eps$, the proof is complete.
\end{proof}

\subsection{Running time analysis}

We first describe some lower-level details of the implementation.
Starting off, the set $\Res$ can be implemented as a collection of balanced binary search trees.
For each vertex $u$, we have a separate balanced BST that stores the set of approximations of $\dist_{G^j}(s,u)$ for different timestamps $j$, sorted in order of increasing $j$.
We can thus perform insertions and queries for the largest/smallest approximation stored for a given range of timestamps in time logarithmic to the number of stored elements.
Later on, we prove that the size of each stored BST is $\O(\log(nW) \log^2(\Delta) / \eps)$.
Therefore, all high-level operations involving $\Res$ take $\O(\log(\log(nW)\log(\Delta)/\eps))$ time.

Secondly, observe that the bottleneck of the cost of a call \textproc{Search} (\Cref{alg:offline}) is running Dijkstra's algorithm in line \ref{line:Dijkstra}.
First, we need to construct the set $\VV$ efficiently in line \ref{line:VV}.
It is too costly to do that naively by iterating through all vertices.
Thankfully, observe that any vertex $v$ may belong to the set $\VV$ only if it belonged to the set $\VV$ of the parent procedure call, if such a call exists.
Therefore, we could use a slightly modified algorithm, where we pass the set $\VV$ to the children calls as $X_0$, so that when constructing respective sets $X$ in children calls, we only iterate through $X_0$.
Then, the total cost of computing all the sets $\VV$ can be easily seen to be proportional to their total size times the cost of checking the condition $\dla(u) > (1 + \epsp)\dgb(u)$.

The key claim that we prove to argue about the efficiency of the proposed preprocessing is that each vertex may belong to the set $\VV$ of only $\O(\log(nW) \log(\Delta) / \epsp)$ recursive calls of \textproc{Search}.

Fix some vertex $u\in V$.
If $u$ belongs to the set $\VV$ inside some $\textproc{Search}$ call, we say that the call is \emph{costly} for $u$.
Our strategy is to bound the number of costly calls at any fixed depth of the recursion.
Formally, we say that the two recursive $\textproc{Search}$ calls are the \emph{$1$-descendants} of the call that invoked them.
Inductively, for any integer $h > 1$, the \emph{$h$-descendants} of some $\textproc{Search}$ call are the $(h - 1)$-descendants of its $1$-descendants.
For convenience, we also say that any call is its own $0$-descendant.
We can now finally state our main lemma in this section.

\begin{restatable}[]{lemma}{restateOfflineComplexity}
	For any $h \ge 0$ and any call $\textproc{Search}(\alpha, \beta)$ that was invoked at some point of the algorithm, the number of its costly $h$-descendants is at most
  \begin{enumerate}[label=${\arabic*}^\circ)$,ref=${\arabic*}^\circ$]
		\item $0$ if $\dla(u) = 0$, \label{case:1}
    \item $1 + \log_{1 + \epsp}(\dla(u))$ if $\dla(u) > 0$ and $\dgb(u) = 0$, \label{case:2}
    \item $\log_{1 + \epsp}\left(\frac{\dla(u)}{\dgb(u)}\right)$ if $\dla(u) > 0$ and $\dgb(u) > 0$. \label{case:3}
	\end{enumerate}
\end{restatable}
\begin{proof}
	\newcommand{\CASE}[1]{\ref{case:#1}}
	We say that the call $\textproc{Search}(\alpha, \beta)$ is the \emph{root} call.
	Case \CASE{1} is obvious: if $\dla(u) = 0$, then the root call is not costly and hence none of its descendants are.

	We prove cases \CASE{2} and \CASE{3} by induction on $h$.
	First, consider $h = 0$.
	If the root falls in case \CASE{2}, then the number of costly $0$-descendants is at most $1$ and the claimed bound is $1 + \log_{1 + \epsp} \dla(u) \ge 1$, since $\dla(u) \ge 1$ --- we use the fact that the edge weights are in $\set{0} \cup [1, W]$.
	If the root falls in case \CASE{3} and is costly, then $\dla(u) > (1 + \epsp)\dgb(u)$ and the bound is $\log_{1 + \epsp} ( \dla(u)/ \dgb(u)) > 1$, which completes the induction base.

  Now assume $h\geq 1$ and that the desired bounds hold for $(h - 1)$-descendants of every $\textproc{Search}$ root call.
  If the root is a leaf call, there is nothing to prove.
	Otherwise, let $\textproc{Search}(\alpha, \gamma - 1)$ and $\textproc{Search}(\gamma + 1, \beta)$ be the $1$-descendants of the root call.
	Also assume that the root is costly, since in the other case, we also have no costly descendants.
	Before the recursive calls, the root computes the estimate $\dg(u)$ and stores it in $\Res$.
	Let $\dgg(u)$ denote the largest estimate of $\dist_{G^j}(s,u)$ stored in $\Res$ for any $j \ge \gamma$ when $\textproc{Search}(\alpha, \gamma - 1)$ is called.
  Similarly,
	let $\dlg(u)$ denote the smallest estimate of $\dist_{G^j}(s,u)$ stored in $\Res$ for any $j \le \gamma$ when $\textproc{Search}(\gamma + 1, \beta)$ is called.
	Observe that we have $\dlg(u) \le \dg(u) \le \dgg(u)$.
	We now have to consider some cases.
	If the root falls into case \CASE{2}, we have two sub-cases:
  \begin{enumerate}[label=$2.{\arabic*}^\circ)$]
		\item $\dg(u) = 0$.
			In that case, the call $\textproc{Search}(\alpha, \gamma - 1)$ must fall into case \CASE{2} and \linebreak $\textproc{Search}(\gamma + 1, \beta)$ into case \CASE{1}.
			The total number of their costly $(h - 1)$-descendants is by induction $1 + \log_{1 + \epsp} \dla(u)$, which bounds the number of costly $h$-descendants of the root.
		\item $\dg(u) > 0$.
			In that case, $\textproc{Search}(\alpha, \gamma - 1)$ falls into case \CASE{3} and $\textproc{Search}(\gamma + 1, \beta)$ into either case \CASE{1} or \CASE{2}.
      Similarly, by $\dlg(u)\leq \dgg(u)$, we get the bound
      \[ \log_{1 + \epsp}\frac{\dla(u)}{\dgg(u)} + 1 + \log_{1 + \epsp}(\dlg(u)) \le 1 + \log_{1 + \epsp}(\dla(u)).\]
	\end{enumerate}
	If the root falls into case \CASE{3}, we also have two sub-cases:
  \begin{enumerate}[label=$3.{\arabic*}^\circ)$]
		\item $\dg(u) = 0$.
			Then, the call $\textproc{Search}(\alpha, \gamma - 1)$ falls into case \CASE{3}, since $\dgg(u) \ge \dgb(u) > 0$, and $\textproc{Search}(\gamma + 1, \beta)$ into case \CASE{1}.
			We get the bound
      $\log_{1 + \epsp} \frac{\dla(u)}{\dgg(u)} \le \log_{1 + \epsp} \frac{\dla(u)}{\dgb(u)}$, as desired.
		\item $\dg(u) > 0$.
			In that case, the call $\textproc{Search}(\alpha, \gamma - 1)$ falls into case \CASE{3}, and $\textproc{Search}(\gamma + 1, \beta)$ into either case \CASE{1} -- where we again get 
			$\log_{1 + \epsp} (\dla(u)/\dgg(u)) \le \log_{1 + \epsp} (\dla(u)/\dgb(u))$ --
			or case \CASE{3}, where the number of costly $h$-descendants is bounded by:
			\[ \log_{1 + \epsp} \frac{\dla(u)}{\dgg(u)} + \log_{1 + \epsp} \frac{\dlg(u)}{\dgb(u)} \le \log_{1 + \epsp} \frac{\dla(u)}{\dgb(u)}. \]
	\end{enumerate}
\end{proof}

Applying the above theorem to the call $\textproc{Search}(0, \Delta)$, we can bound the total number of costly calls for a given vertex $u$ by
$\left(1 + \log_{1 + \epsp}\left(nW\right)\right) \cdot \log_2 \Delta$, which after substituting 
$\epsp=\Theta(\eps/\log\Delta)$ is 
$\O(\log(nW) \log^2(\Delta) / \eps)$.
The total cost of the $\textproc{Dijkstra}$ invocations (line \ref{line:Dijkstra} of \Cref{alg:offline}), combined with the total number of $\textproc{Search}$ calls, is
$\O(\Delta \log(\Delta) + m \log(n)\log(nW) \log^2(\Delta) / \eps).$

Recall from~\Cref{sec:all-pairs} that for large values of $\Delta = \Omega(m \log(W)/\eps)$, we can reduce our problem to the case where we have at most $\O(m \log(W)/\eps)$ weight decreases in total by filtering the updates in $\O(\Delta)$ time at the start.
The time complexity thus becomes
$\O(m \log(n)\log(nW)\log^2(\Delta) / \eps + \Delta)$.

\section{Randomized reset}\label{a:reset}
This section is devoted to proving~\Cref{t:randomized-reset}.
Below we state the procedure $\textsc{Randomized-Reset}$ that achieves the goal if called with parameter $\bal$.
After the procedure completes, we simply reset the rank offset $\roff$ to $\bal$.
Provided that the procedure makes every vertex $\bal$-certified, any particular vertex $v$ is clearly
$(\roff+\rank(v))$-certified afterwards.

\begin{algorithm}[h!]
  \caption{\textsc{Randomized-Reset}($\lambda$)}\label{alg:batch-update}
  \begin{algorithmic}[1]
    \State $Z:=\emptyset$
    \For{$i=1,\ldots,\lceil c\cdot (25\ell/\bal)\log{n}\rceil$}
      \State sample an integer $j\in \{0,1,\ldots,\lfloor \bal/8\rfloor\}$ uniformly at random
      \State $Y:=\{v:(1+\epsp)^{j\cdot 8\ell/\bal}\leq d(v)\leq (1+\epsp)^{(j+2)\cdot 8\ell/\bal}\}$ \Comment{compute $Y$ only if $\deg(Y)=\O(m\ell/\bal^2)$}
      \If{$\deg(Y)\leq 400m\ell/\bal^2$}
        \State $Z:=Z\cup Y$
      \EndIf
    \EndFor
    \State \Call{Propagate}{$Z$}
  \end{algorithmic}
\end{algorithm}
Note that in $\textsc{Randomized-Reset}$,
the sum of degrees of the vertices in the sampled set $Z$ is $\O(m\ell^2\log{n}/\bal^3)$.
Hence the total cost of $\PD(Z)$ that cannot be charged to estimate drops
is $\O(m\ell^2\log^2(n)/\bal^3)=\O(m\log^2(nW)\log^2(n)/(\bal^3\epsp^2))$.

The following lemma proves the correctness of $\textsc{Randomized-Reset}$.
\begin{lemma}\label{l:batch-insert-correct}
  After $\textsc{Randomized-Reset}(\lambda)$ completes, every $v\in V$ is $\lambda$-certified
  w.h.p.
\end{lemma}

\begin{proof}
If $P$ is a simple path, we will write $P[x\to y]$ to denote the subpath of $P$ that is an $x\to y$ path.
  Suppose the lemma does not hold and let $u\in V$ be a vertex that is not $\lambda$-certified.
  There thus exists a shortest path $P=u\to t$ in $G[V\setminus \{s\}]$ such that
  $d(t)>(1+\epsp)^\bal(d(u)+\wei(P))$ and $|P|$ is minimal.

  For any~$v$, let $d'(v)$ be the estimate $d(v)$ before running $\textsc{Randomized-Reset}$. 
  By~\Cref{slack-inv}, we have $d'(v)\leq (1+\epsp)(d'(u)+\wei(e))$ for any edge $uv=e\in E$.
  Since the estimates only decrease, we also have $d'(t)>(1+\epsp)^\bal(d(u)+\wei(P))$.

  Define vertices $u=u_0,u_1,\ldots,u_k=t$ on the path $P$ as follows. Let $u_0,\ldots,u_k$
all lie on~$P$ in this order and for all $i=1,\ldots,k$, $u_i$ is the \emph{last} vertex on~$P$ following $u_{i-1}$ such that the subpath $P_i=u_{i-1}\to u_i$ satisfies
  \[ d'(u_i)\leq (1+\epsp)(d'(u_{i-1})+w(P_i)). \]
The vertex $u_i$ is well-defined by~\Cref{slack-inv}.
Note that one can chain the above inequalities for $i=1,2,\ldots,k$ and get
  \[d'(t)=d'(u_k)\leq (1+\epsp)^k(d'(u_0)+w(P))=(1+\epsp)^k\cdot \dist(s,t).\]
  Hence, by our assumption, we conclude $k>\bal$.
  Observe that, by the definition of $u_{i}$ and $u_{i+1}$, for any $1\leq i\leq k-1$, we have:
  \begin{equation}\label{eq:gap2}
    d'(u_{i+1})>(1+\epsp)(d'(u_{i-1})+w(P_iP_{i+1}))\geq (1+\epsp)d'(u_{i-1})
  \end{equation}
  and
  \begin{align*}
    d'(u_{i+1})&\leq (1+\epsp)(d'(u_i)+\wei(P_{i+1}))\\
               &\leq (1+\epsp)((1+\epsp)(d'(u_{i-1})+\wei(P_{i}))+\wei(P_{i+1}))\\
               &\leq (1+\epsp)^2(d'(u_{i-1})+\wei(P_{i}P_{i+1})).
  \end{align*}

  Note that there exists at most $\frac{1}{3}\bal$ indices $i$ such that
  \[ d'(u_{i-1})+\wei(P_iP_{i+1})>(1+\epsp)^{6\ell/\bal}\cdot d'(u_{i-1}), \]
  since otherwise we would be able to chain the above inequalities for some
  $\frac{1}{6}\bal$ distinct indices $i$ using~\eqref{eq:gap2}
  and obtain $d'(u_l)>(1+\epsp)^{\ell}$ for some $l$,
  which is impossible by $d'(u_l)\leq nW$.
  
  Let thus $X$, $|X|\geq k-\bal/3\geq \frac{2}{3}\bal$, be the set of indices $i=2,\ldots,k-1$ such that
  \[ d'(u_{i-1})+\wei(P_iP_{i+1})\leq (1+\epsp)^{6\ell/\bal}\cdot d'(u_{i-1}). \]
  For each $i\in X$, we have:
  \[ 
  d'(u_{i+1})\leq (1+\epsp)^2(d'(u_{i-1})+\wei(P_{i}P_{i+1})) \leq
  (1+\epsp)^{6\ell/\bal+2}\cdot d'(u_{i-1})\leq (1+\epsp)^{8\ell/\bal}\cdot d'(u_{i-1}). \]
  
  Now consider a family of $\bal$ intervals $\mathcal{I}$ of the form
  \[[(1+\epsp)^{j\cdot 8\ell/\bal}, (1+\epsp)^{(j+2)\cdot 8\ell/\bal}],\]
  for integer values $j=0,\ldots,\lfloor \bal/8\rfloor$. Note that for $i\in X$, an interval
  \[ J_i := [d'(u_{i-1}),(1+\epsp)^{8\ell/\bal}\cdot d'(u_{i-1})]\supseteq [d'(u_{i-1}),d'(u_{i+1})] \]
  is contained in some interval from $\mathcal{I}$.
  
  On the other hand, observe that for any $I\in \mathcal{I}$, we have
  \[|I\cap d'(X)|\leq 64 \ell/\bal.\]
  Indeed, if $s\in X$ is minimum such that $d'(s)\in I$, then by chaining~\eqref{eq:gap2} for
  $i=s+1,s+3,\ldots$, we obtain:
  \[ d'(u_{s+2\cdot (2\cdot 8\lceil \ell/\bal\rceil-1) })>(1+\epsp)^{2\cdot 8\ell/\bal-1}\cdot d'(u_{s}),\]
  and thus for all $i\geq s+64\ell/\bal\geq s+32\lceil \ell/\bal\rceil =s+2\cdot 2\cdot 8\lceil\ell/\bal\rceil$ we have $d'(u_i)>(1+\epsp)^{2\cdot 8\ell/\bal}d'(u_s)\notin I$.
  We conclude that at least $(\bal/64\ell)|X|\geq \bal^2/100\ell$ distinct intervals from $\mathcal{I}$ contain some interval $J_i$, where $i\in X$.

  For $I\in \mathcal{I}$, define
  \[ \deg(I)=\sum_{v:d'(v)\in I}\deg_G(v). \]
  Since each value $d'(v)$ is contained in at most two intervals from $\mathcal{I}$, 
  we have \[\sum_{I\in\mathcal{I}} \deg(I)\leq 2m.\]
  Consequently, for at most $\bal^2/200\ell$ intervals $I$ from $\mathcal{I}$
  we have $\deg(I)\geq 400m\ell/\bal^2$.
  Since \linebreak $(\bal^2/200\ell)/(\bal/8)=\bal/25\ell$, it follows that at least an $\bal/25\ell$-fraction of intervals $I\in\mathcal{I}$ satisfy both:
  \begin{enumerate}
    \item $\deg(I)\leq 400m\ell/\bal^2$,
    \item There exists $2\leq i<k$ such that $[d'(u_{i-1}),d'(u_{i+1})]\subseteq I$.
  \end{enumerate}
  Hence, a subset $Z\subseteq V$ obtained in Algorithm~\ref{alg:batch-update} by unioning $c\cdot (25\ell/\bal)\log{n}$ uniformly sampled (small-degree) intervals from~$\mathcal{I}$, for some constant $c\geq 1$,
  has total degree $\O(m\ell^2\log(n)/\bal^3)$, and contains, 
  with high probability, all vertices with their respective $d'(\cdot)$ estimates in the interval
  $[d'(u_{i^*-1}),d'(u_{i^*+1})]$ for some $i^*\leq \bal$.
  
  To finish the proof, we will argue that assuming $Z$ has the desired properties (which is the case whp.),
  running $\PD(Z)$ ensures that
  \begin{equation}\label{eq:all-in-queue}
    d(t)\leq d(u_{i^*-1})+\wei(P_{i^*}\ldots P_k).
  \end{equation}
  To see that this is a contradiction, note that the $u\to u_{i^*-1}$ path $P_1\ldots P_{i^*-1}$ has less edges than
  $P$ and thus $d(u_{i^*-1})\leq (1+\epsp)^\bal\cdot (d(u)+\wei(P_1\ldots P_{i^*-1}))$
  by the minimality of $P$. Therefore:
  \begin{align*}
    d(t)&\leq d(u_{i^*-1})+\wei(P_{i^*}\ldots P_k)\\
    &\leq (1+\epsp)^\bal\cdot (d(u)+\wei(P_1\ldots P_{i^*-1}))+\wei(P_{i^*}\ldots P_k)\\
    &\leq (1+\epsp)^\bal \cdot (d(u)+\wei(P)).
  \end{align*}
  This contradicts our choice of $P$.

  Let $P':=P_{i^*}\ldots P_k=v_0v_1\ldots v_p$, where $v_0=u_{i^*-1}$, $v_p=t$ and $v_{y}=u_{i^*+1}$.
  To prove~\eqref{eq:all-in-queue}, 
  we will inductively show that for all $j=0,\ldots,p$, $v_j$ is inserted in the queue $Q$ at some point of the call $\PD(Z)$, or, in other words, $V(P')\subseteq \Vtouch$.
  By~\Cref{lem:no_slack}, it will follow
  that every edge of~$P'$ is relaxed afterwards, and thus~\eqref{eq:all-in-queue} holds.
  
  For $j=0$, the claim holds since $v_0=u_{i^*-1}\in Z$. Suppose $1\leq j<y$ and that the claim holds for all smaller $j$.
 First, note that $d'(v_j)\geq d'(v_0)$, since by the definition of $u_{i^*-1}$, we have
\begin{align*}
  d'(v_j)&>(1+\eps)(d'(u_{i^*-2})+\wei(P[u_{i^*-2}\to v_j]))\\
  &\geq (1+\eps)(d'(u_{i^*-2})+\wei(P[u_{i^*-2}\to u_{i^*-1}])\\
  &\geq d'(u_{i^*-1})
  =d'(v_0).
\end{align*}
 Moreover, if $d'(v_j)<d'(v_y)$, then $d'(v_j)$ is also in~$Q$ initially since all vertices with their $d'$ values in $[d'(v_0),d'(v_y)]$ are included in $Z$.
 We may thus assume $d'(v_j)\geq d'(v_y)>(1+\epsp)(d'(v_0)+\wei(P_{i^*}P_{i^*+1}))$.
 Since $j<y$, we have \[d'(v_0)+w(P_{i^*}P_{i^*+1})\geq d'(v_0)+w(P[v_0\to v_{j-1}])+\wei(v_{j-1}v_j).\]
 But by the inductive assumption $d(v_{j-1})\leq d'(v_0)+\wei(P[v_0\to v_{j-1}])$, which in turn implies that
\begin{align*}
  d'(v_j)&>(1+\epsp)(d'(v_0)+\wei(P[v_0\to v_{j-1}])+\wei(v_{j-1}v_j))\geq (1+\epsp)(d(v_{j-1})+\wei(v_{j-1}v_j))
\end{align*}
holds when $v_{j-1}$ is processed.
 As a result, $v_j$ will be pushed to $Q$ when processing the edge $v_{j-1}v_j$ or earlier.

 Now, suppose $j\geq y$. Recall that $d'(v_j)>(1+\epsp)(d'(v_0)+\wei(P[v_0\to v_j]))$. By the inductive assumption, $v_{j-1}$ is
 put into $Q$ at some point and $d(v_{j-1})\leq d'(v_0)+\wei(P[v_0\to v_{j-1}])$.
 As a result, when processing $v_{j-1}$, we have $d'(v_j)>(1+\epsp)(d(v_{j-1})+\wei(v_{j-1}v_j))$, which ensures
 $v_j$ is pushed to~$Q$ as well.
\end{proof}

\clearpage
\bibliographystyle{alpha}
\bibliography{references}

\end{document}